\documentclass[11pt]{article}
\usepackage{amsmath}
\usepackage{graphicx}
\usepackage{enumerate}
\usepackage{natbib}
\usepackage{url} 
\usepackage{amsfonts, amssymb, amsthm, graphicx, algorithmic,
  algorithm, fullpage}
\usepackage{color}

\newcommand{\real}{{\mathbb R}}
\def\st{\mathrm{\quad s.t.\quad}}
\def\half{\frac{1}{2}}

\def\X{\mathbf X}
\def\S{\mathbf S}
\def\bfA{{\bf A}}
\def\calA{\mathcal A}
\def\B{{\bf B}}
\def\C{{\bf C}}
\def\D{{\bf D}}
\def\U{{\bf U}}
\def\bSigma{\boldsymbol\Sigma}
\def\hSigma{\hat{\bSigma}}
\def\tSigma{\tilde{\bSigma}}
\def\bbSigma{\bar{\bSigma}}
\def\I{{\mathcal I}}
\def\diag{\mathrm{\bf diag}}
\def\diam{\mathrm{\bf diam}}
\def\supp{\mathrm{supp}}
\def\ggbgb{\text{GGB-global}}
\def\ggbvb{\text{GGB-local}}

\newtheorem{defn}{Definition}
\newtheorem{ass}{Assumption}
\newtheorem{lem}{Lemma}

\newtheorem{thm}{Theorem}
\newtheorem{theorem}{Theorem}
\newtheorem{prop}{Proposition}
\newtheorem{cor}{Corollary}

\newcommand{\V}[1][b]{{\bf V}^{(#1)}}

\newcommand{\hV}[1][b]{\hat{{\bf V}}^{(#1)}}

\newcommand{\shV}{\{\hV\}}
\newcommand{\sbV}{\{\bV\}}

\newcommand{\sV}[1][b]{\{\V[#1]\}}

\newcommand{\bV}[1][b]{\bar{{\bf V}}^{(#1)}}

\renewcommand{\P}{\mathbb P}
\def\pro{(\bSigma^*_{jj}\bSigma^*_{kk})^{1/2}}
\def\hpro{(\S_{jj}\S_{kk})^{1/2}}
\def\cA{{\mathcal A}}
\def\corr{{\bf R}^*}
\def\hcorr{{\hat{\bf R}}}
\def\trace{\mathrm{trace}}

\begin{document}

\bibliographystyle{agsm}

  \title{\bf Graph-Guided Banding of the
Covariance Matrix}
  \author{Jacob Bien\thanks{
    The author thanks Christian M\"uller for a useful conversation and acknowledges
the support of the NSF grant DMS-1405746 and NSF CAREER Award DMS-1748166.}\hspace{.2cm}\\
    Data Sciences and Operations\\
    Marshall School of Business\\
University of Southern California}
\date{}
  \maketitle

\begin{abstract}
Regularization has become a primary tool for developing
  reliable estimators of the covariance matrix in high-dimensional settings.  To curb the
  curse of dimensionality, numerous
  methods assume that the population
  covariance (or inverse covariance) matrix is sparse, while making no
  particular structural assumptions on the desired pattern of sparsity.
  A highly-related, yet complementary, literature studies the specific
  setting in which the measured
  variables have a known ordering, in which case a banded population
  matrix is often assumed.  While the banded
  approach is conceptually and computationally easier than asking for
  ``patternless sparsity,'' it is only
  applicable in very specific situations (such as when data are
  measured over time or one-dimensional space).  This work proposes a
  generalization of the notion of bandedness that greatly expands the
  range of problems in which banded estimators apply.

We develop convex regularizers occupying the broad middle ground
between the former approach of ``patternless sparsity'' and the latter
reliance on having a known ordering.  Our framework defines bandedness
with respect to a known graph on the measured variables.  Such a graph
is available in diverse situations, and we provide a theoretical, computational,
and applied treatment of two new estimators.  An R package, called
{\tt ggb}, implements these new methods.

\end{abstract}

\noindent


\section{Introduction}
\label{sec:introduction}

Understanding the relationships among large numbers of
variables is a goal shared across many scientific areas.
Estimating the covariance matrix is perhaps the most basic step
toward understanding these relationships, and yet even this task
is far from simple when sample sizes are small relative to the number
of parameters to be estimated.  High-dimensional covariance
estimation is therefore an active area of research.  Papers in
this literature generally describe a set of structural assumptions
that effectively reduces the dimension of the parameter space to make
reliable estimation tractable given the available sample size.  Some
papers focus on eigenvalue-related structures and shrinkage 
\citep[e.g.,][among many others]{johnstone2001distribution,fan2008high,won2013condition,donoho2013optimal,
Ledoit04}; others assume that correlations are
constant or blockwise constant (such as in random effects models);
 and a very large
number of papers introduce sparsity assumptions on the covariance matrix
(or its inverse).  Still other papers suggest combinations of these
structures \citep{chandrasekaran2010latent,Bien11spcov,luo2011recovering,Rothman12,fan2013large,Liu13,xue2012positive}.

Given the large number of papers focused on the sparsity assumption,
it is perhaps surprising to note that nearly all such papers fall into
one of just two categories:
\begin{enumerate}
\item Methods that place {\bf no assumption on the pattern of
    sparsity} but merely require that the {\em number} of
nonzeros be small.  This assumption is placed on either the inverse
covariance matrix
\citep{Dempster72,MB2006,Yuan07,Rothman08,Banerjee08,glasso,peng2009partial,yuan2010high,cai2011constrained,khare2014convex}
or on the covariance matrix itself
\citep{Bickel08a,rothman2009generalized,Lam09,Bien11spcov,Rothman12,xue2012positive,Liu13}.

\item Methods that {\bf assume the variables have a known ordering}.
  These methods generally assume a banded (or tapered) structure for
  the covariance matrix, its inverse, or the Cholesky factor \citep{wu2003nonparametric,Huang06,Levina08,Bickel08band,wu2009banding,Rothman10,Cai10,Cai12,Bien15convex,yu2016learning}.
\end{enumerate}
We can view these two categories of sparsity as extremes on a
spectrum.  Methods
in the first category tackle a highly challenging problem that is
widely applicable.  Being completely neutral to the pattern of zeros
leads to an enormous number of sparsity patterns to choose from: there
are ${\smash{\binom{\binom{p}{2}}{s}}}$ for an $s$-sparse model.  By constrast,
methods in the second category address a more modest problem that
is much less widely applicable.  Using a simple banding estimator, there are
only $p$ sparsity patterns to choose from (one for each bandwidth) and
there are few examples other than data collected over time in which
such estimators apply.

This paper introduces a new category of sparsity patterns that aims to
bridge the wide divide between these two categories of sparsity
pattern.  The goal is to produce an estimator that combines the
problem-specific, targeted perspective of banded
covariance estimation with the wide applicability of the patternless
sparsity approaches.  The key idea is to generalize the notion of
banding beyond situations in which the variables have a known
ordering.  In particular, we replace the assumption that the variables ``have a known
ordering'' with the assumption that they ``lie on a known graph.''
The latter is a generalization because 
variables having a known ordering can be expressed as variables lying on a path graph.

To be precise, we will assume that we observe $n$ independent copies
of a random vector $\X\in\real^p$, where the $p$ variables are
associated with a known graph $G=([p],E)$.  We will refer to $G$ as
the {\em seed graph}.
\begin{figure}
  \centering
  \includegraphics[width=0.3\linewidth]{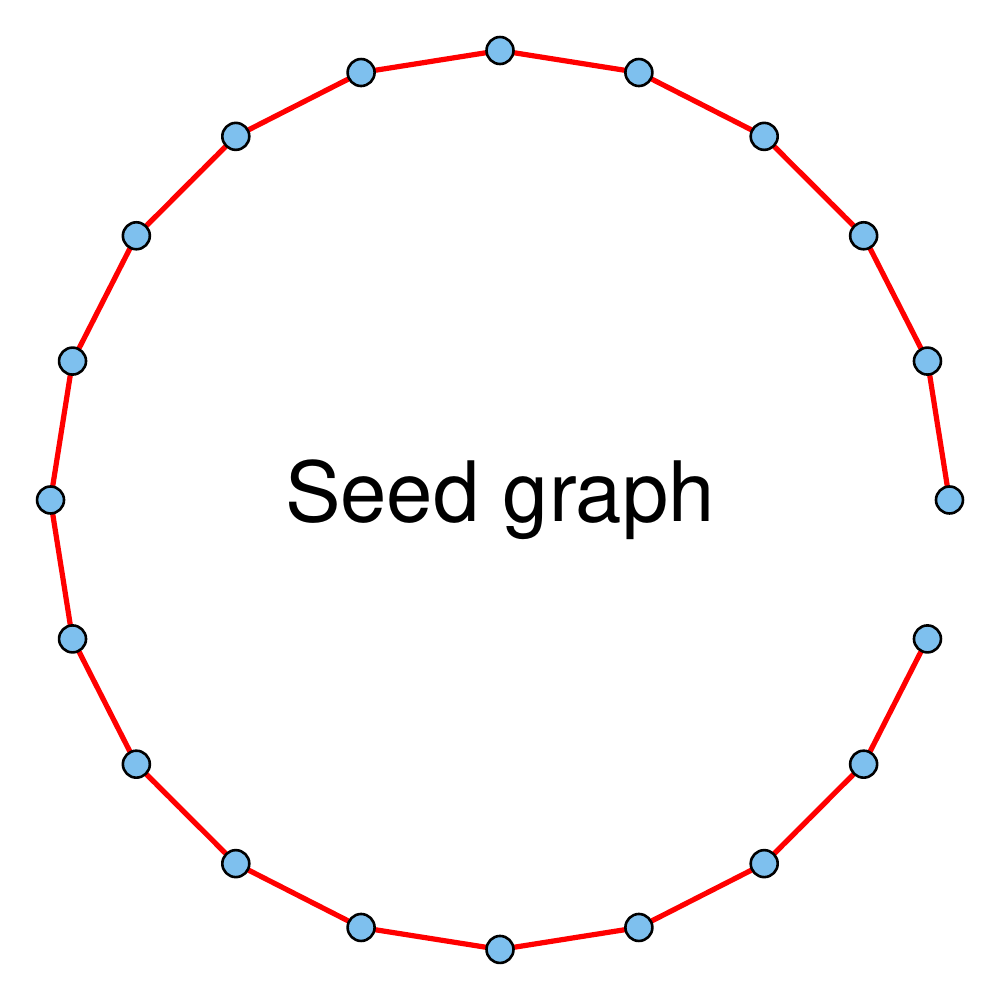}
  \includegraphics[width=0.3\linewidth]{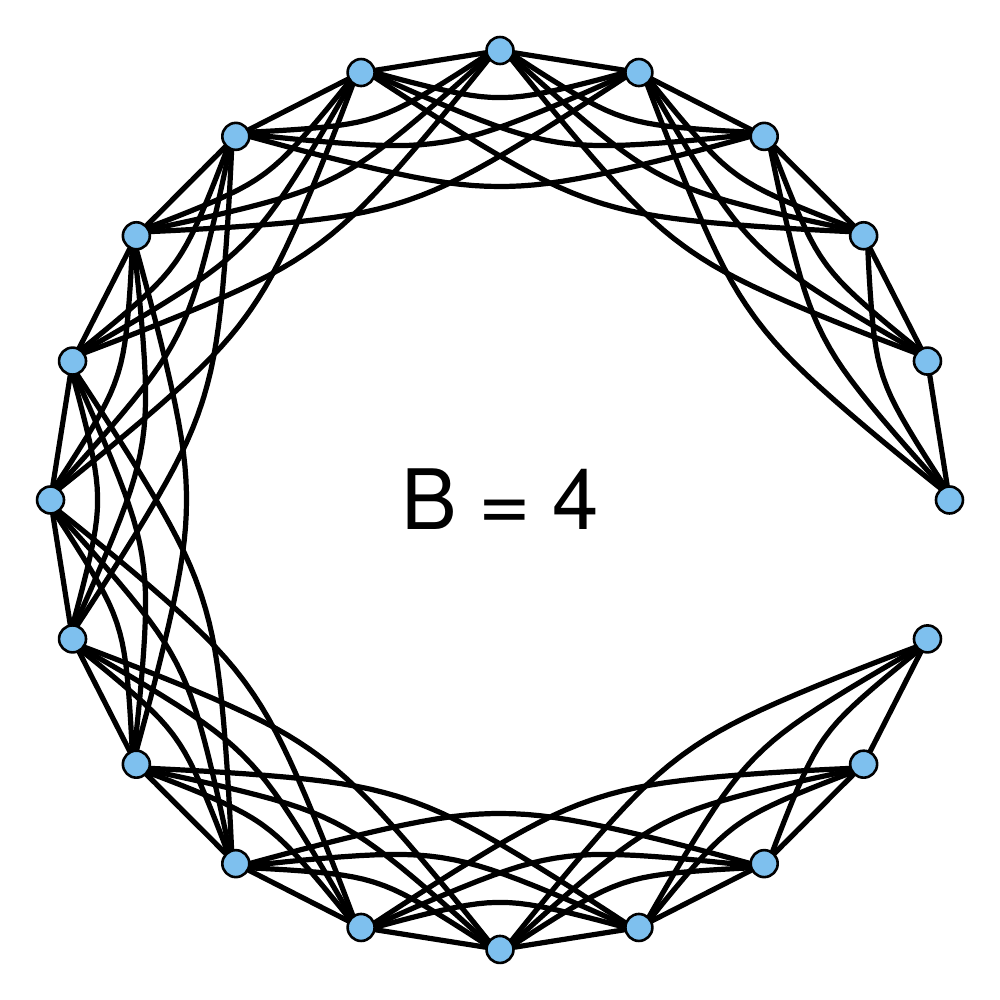}
  \includegraphics[width=0.3\linewidth]{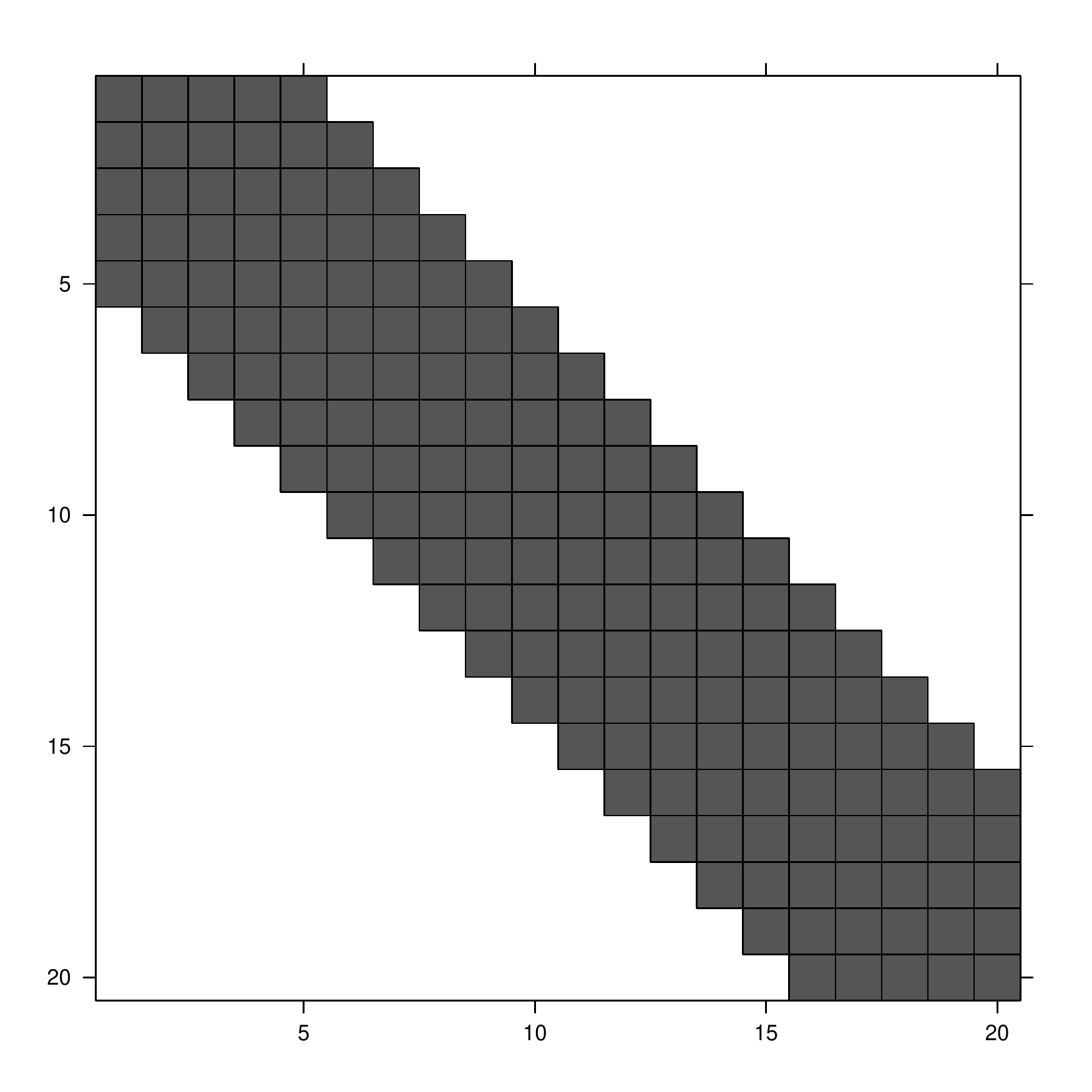}
  \caption{\small\em A matrix with bandwidth $B=4$ (right panel) can be thought
    of as the adjacency matrix of the $B$-th power (center panel) of a
    path graph
    $G$ (left panel).  To form the $B$-th power of a graph $G$, one connects any pair of nodes that are within $B$ hops of each other on $G$.}
  \label{fig:banded}
\end{figure}
The right panel of Figure \ref{fig:banded} shows a traditional
$B$-banded matrix.  The term ``banded'' refers to the diagonal strip
of nonzero elements that are within some distance $B$ (called the ``bandwidth'') of
the main diagonal.  The middle panel of Figure \ref{fig:banded} shows
the graph with this adjacency matrix. Observing that this graph,
denoted $G^B$, is the
$B$-th power of the path graph shown in the left panel of Figure \ref{fig:banded} suggests the following
generalization of the term ``banded.''  We will use 
$d_G(j,k)$ to denote the length of the shortest path between nodes $j$
and $k$ on $G$.
\begin{defn}
  We say that a matrix $\bSigma$ is {\bf $B$-banded with respect to a graph
    $G$} if $\supp(\bSigma)\subseteq E(G^B)$, that is $\bSigma_{jk}=0$
  if $d_G(j,k)> B$.
\label{def:banded}
\end{defn}
\citet{bickel2012approximating} have a similar notion, for a metric
$\rho$, which they call
{\em $\rho$-generalized banded operators}.

Figure \ref{fig:ggb} shows the sparsity pattern of a $4$-banded matrix
with respect to a seed graph.  The bandwidth $B$ can vary from $B=0$,
which corresponds to $\bSigma$ being a diagonal matrix, to $B$ being
the diameter of $G$, which corresponds to $\bSigma$ being a dense
matrix when $G$ is a connected graph (when $G$ is not connected,
the densest $\bSigma$ attainable is, up to permutation of the variables, a block-diagonal matrix
with completely dense blocks).
\begin{figure}
  \centering
  \includegraphics[width=0.3\linewidth]{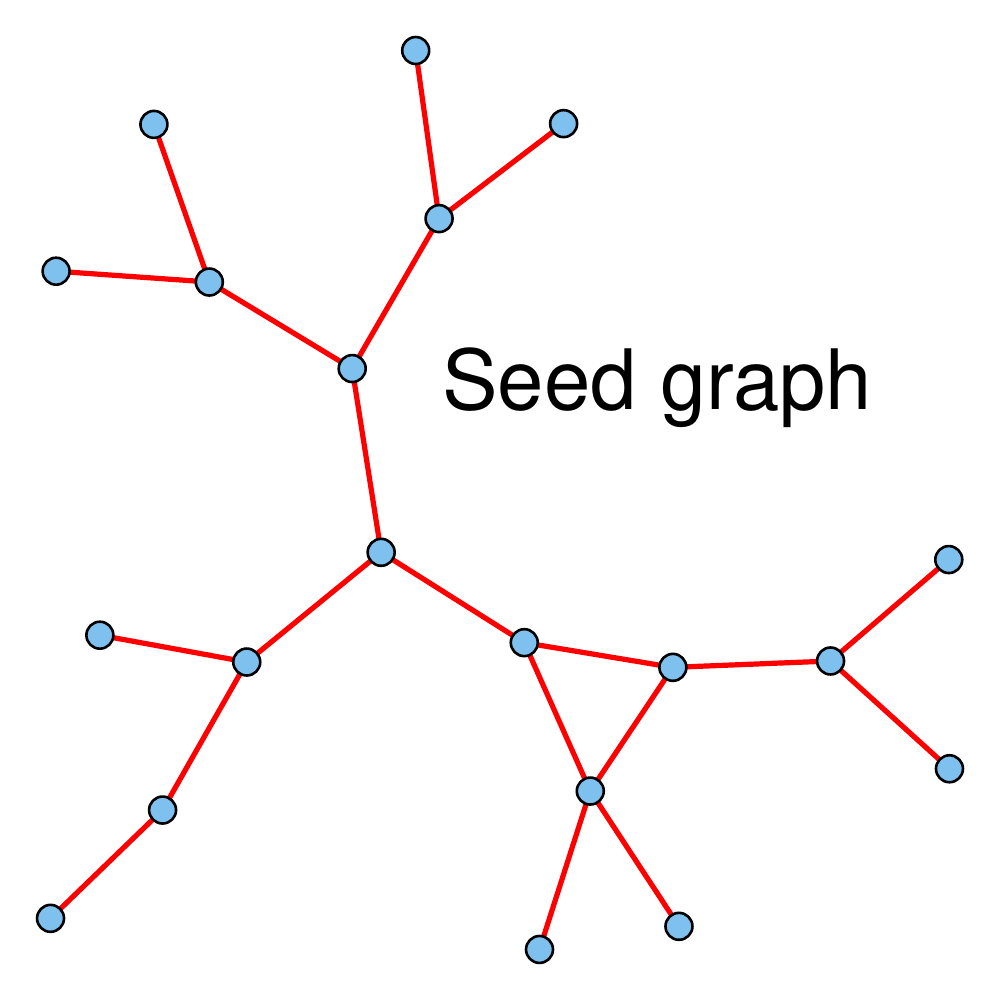}
  \includegraphics[width=0.3\linewidth]{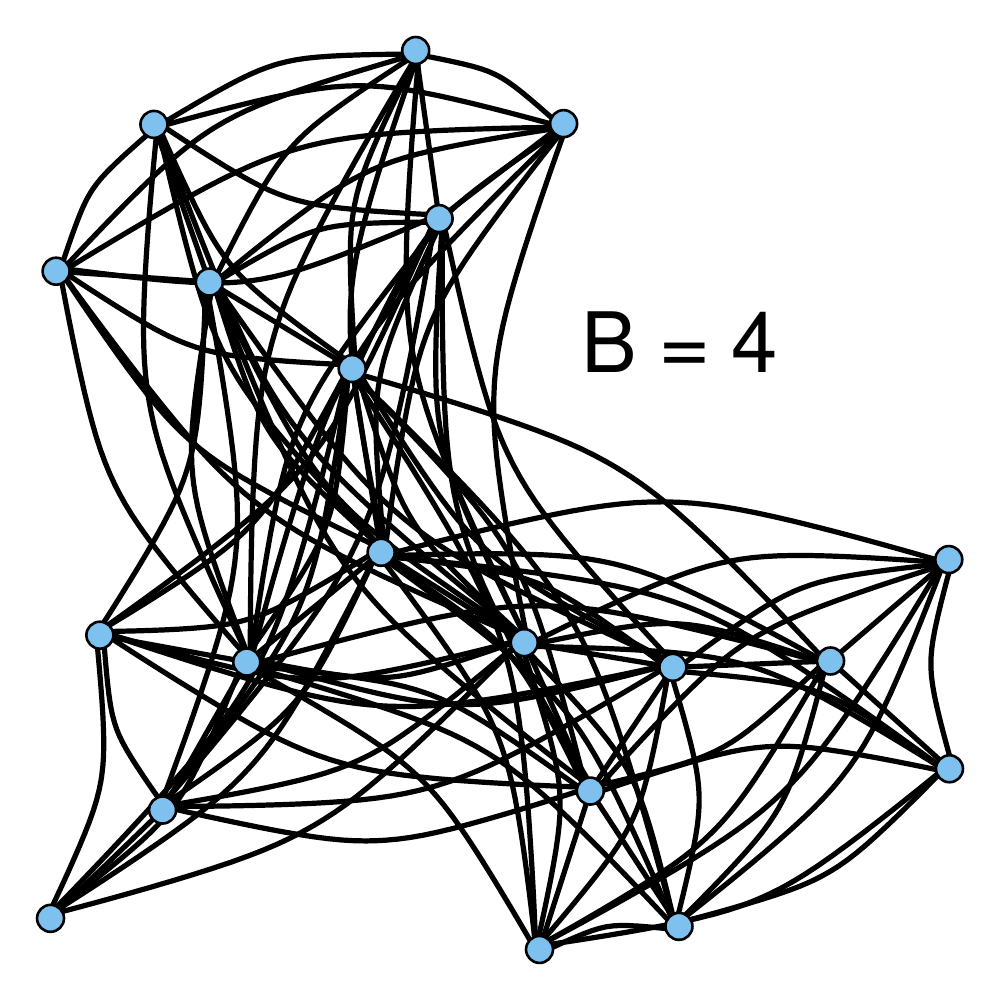}
  \includegraphics[width=0.3\linewidth]{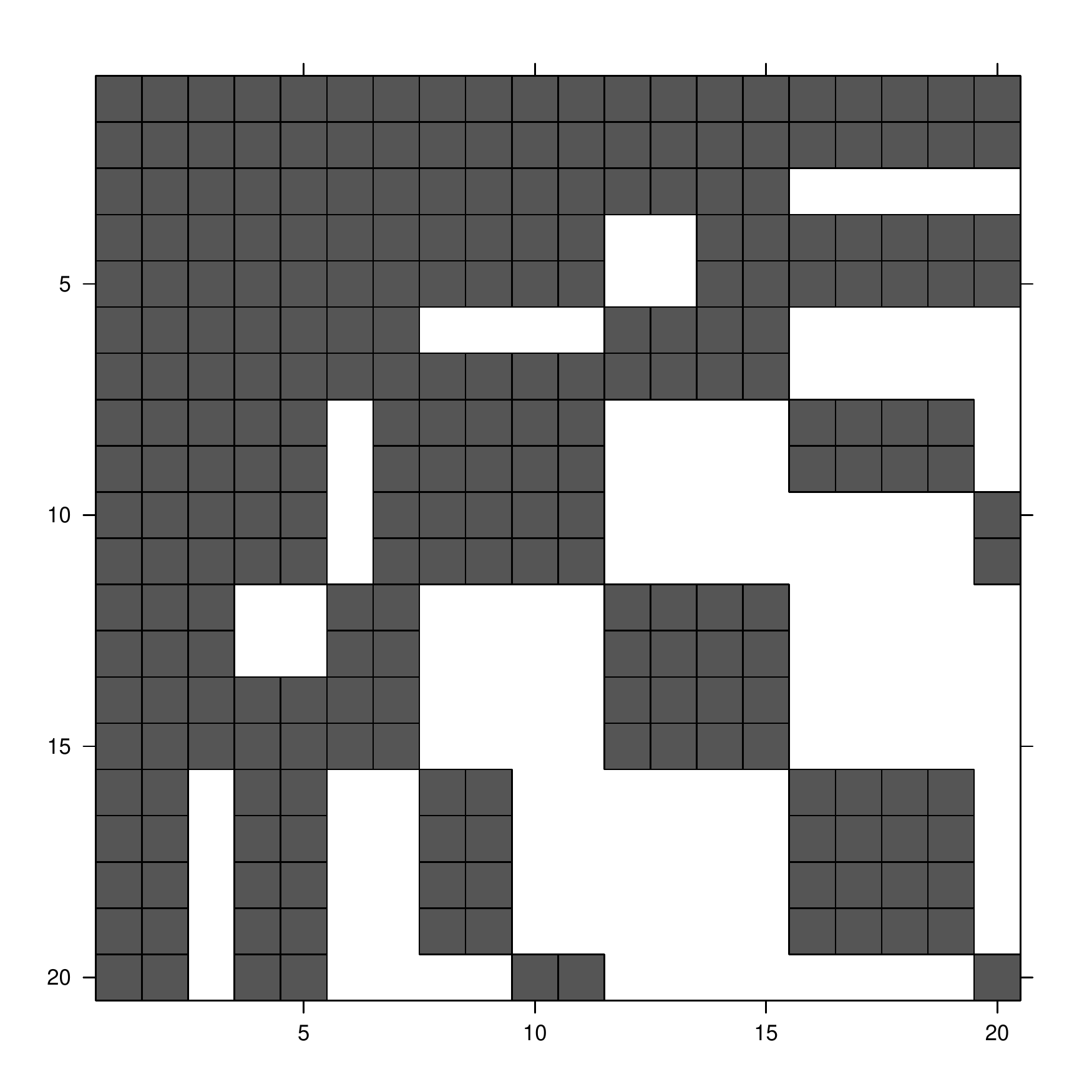}
  \caption{\small\em Given a seed graph $G$ (left panel), one forms $G^B$ ($B=4$
    case shown in center panel) by connecting any pair of nodes within
    $B$ hops of each other in $G$.  Definition \ref{def:banded} refers
    to the resulting adjacency matrix (right panel) as $B$-banded with
  respect to the graph $G$.}
  \label{fig:ggb}
\end{figure}
Despite the simple interpretation of taking powers of a seed graph,
the resulting sparsity patterns can be quite complex and varied
depending on the seed graph.  
There are many applications in which measured variables lie on a known
graph that would be relevant to incorporate when estimating the
covariance matrix $\bSigma$:
\begin{itemize}
\item {\em Image data:} Consider a $p_1\times p_2$ pixeled image.  The
  $p=p_1p_2$ pixels lie on a two-dimensional lattice graph.  Nearby
  pixels would be expected to be correlated.
\item {\em fMRI data:} One observes three-dimensional images of the
  brain over time.  The {\em voxels} (three-dimensional pixels) lie on
  a three-dimensional lattice.
\item {\em Social network data:} A social media company wishes to
  estimate correlations between users' purchasing behaviors (e.g., do
  they tend to buy the same items?).  The social network of users is
  the known seed graph.
\item {\em Gene expression data:} Genes that encode proteins that are
  known to interact may potentially be more likely to be co-expressed
  than genes whose proteins do not interact. Biologists have constructed protein
  interaction networks \citep{das2012hint}, which may be taken as seed
  graphs, leveraging protein interaction data for better estimation of the covariance matrix of gene expression data.
\end{itemize}
In each example described above, the naive approach would be to
vectorize the data and ignore the graph structure when estimating the
covariance matrix.  However, doing so when the truth is banded with respect to the known
graph would be unnecessarily working in the patternless
sparsity category when in fact a much simpler approach would work.

In some situations, a known seed graph is available, but the
assumption of a single bandwidth $B$ applying globally is
unrealistic.  We can therefore extend the definition of bandedness
with respect to a graph as follows:
\begin{defn}
  We say that a matrix $\bSigma$ is {\bf $(B_1,\ldots,B_p)$-banded with respect to a graph
    $G$}, if $\bSigma_{jk}=0$
  for all $(j,k)$ satisfying $d_G(j,k)> \max\{B_j,B_k\}$.
\label{def:banded-vb}
\end{defn}
In words, this definition says that each variable has its own
bandwidth, which defines a neighborhood size, and it cannot be correlated with
any variables beyond this neighborhood.  For example, on a
social network, the bandwidth can be thought informally as having to
do with the influence of a person.  A very influential friend might
be correlated with not just his friends but with his friends' friends (this
would correspond to $B_j=2$).  By contrast, a ``lone wolf'' might be
uncorrelated even with his own friends ($B_j=0$).  Figure
\ref{fig:ggb-vb} shows an example of a sparsity pattern emerging from
the same seed graph as in Figure \ref{fig:ggb} but with variable
bandwidths.

\begin{figure}
  \centering
  \includegraphics[width=0.3\linewidth]{figs/exampleseed.pdf}
  \includegraphics[width=0.3\linewidth]{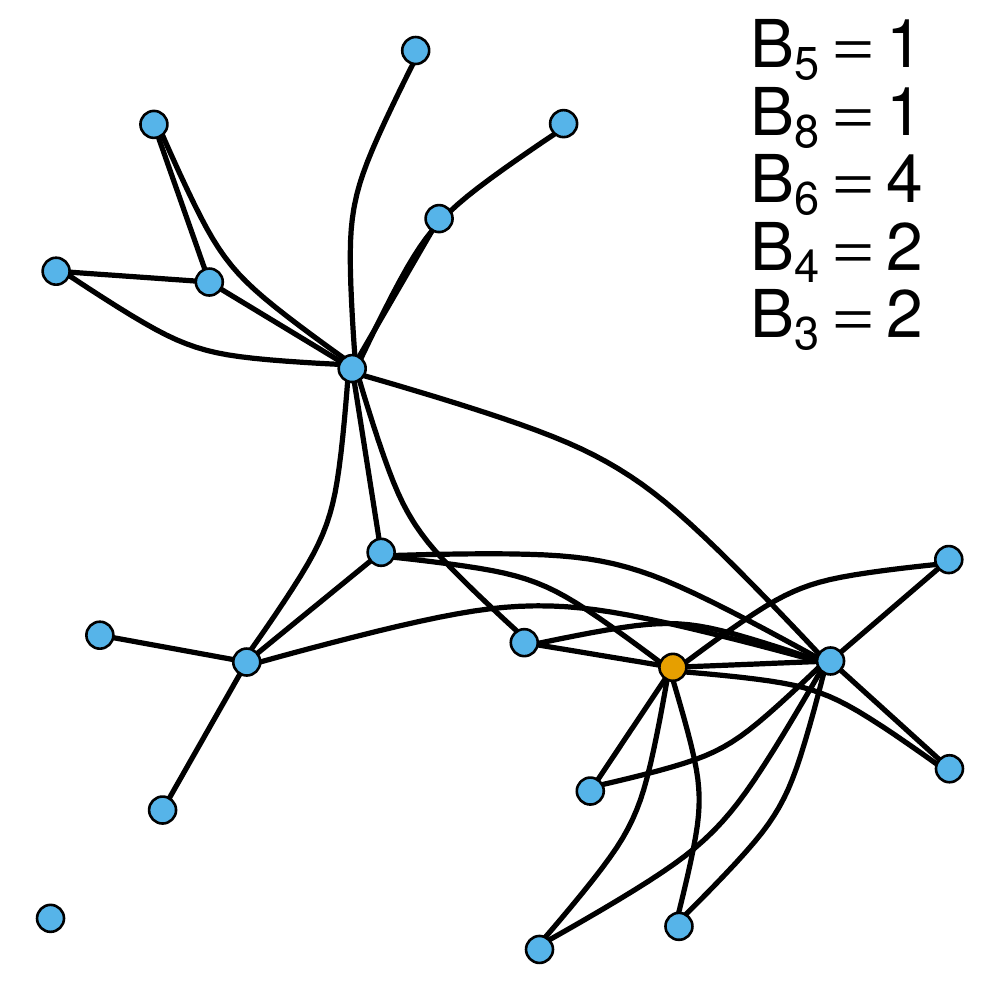}
 \includegraphics[width=0.3\linewidth]{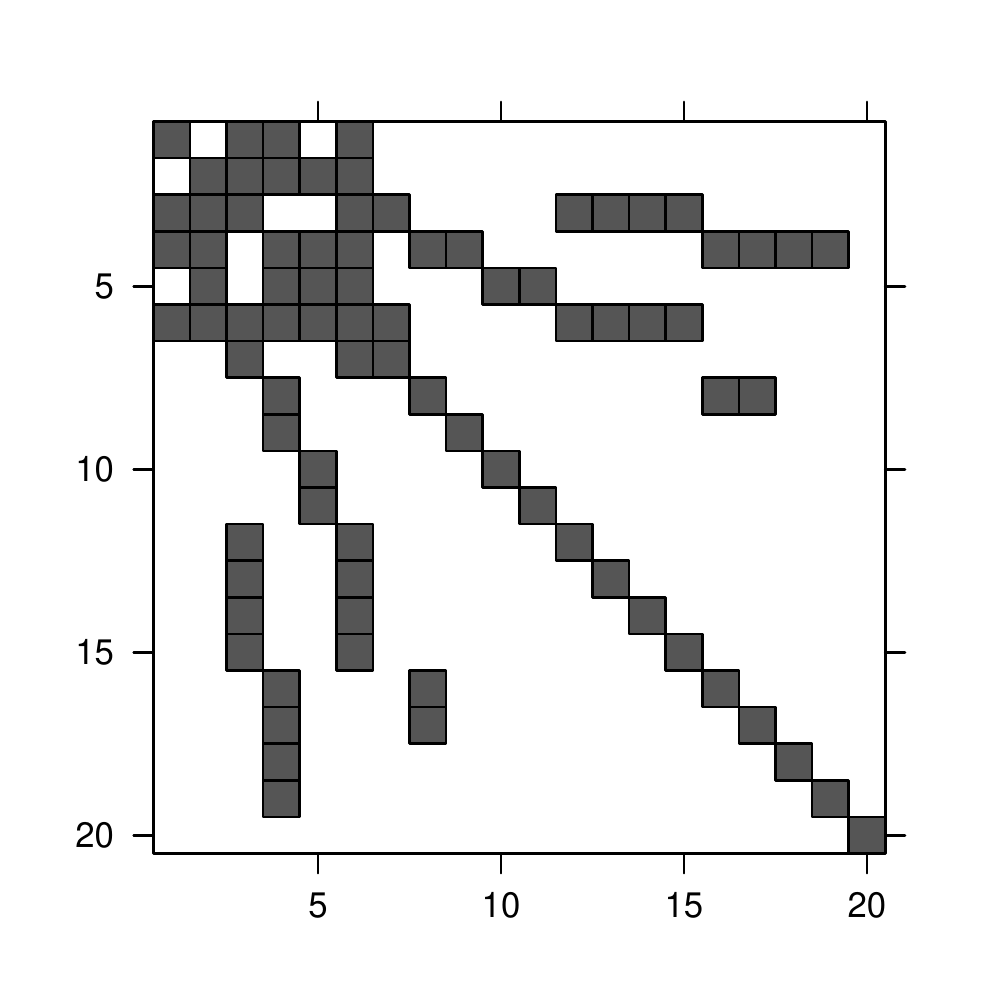}
  \caption{\em Given a seed graph $G$ (left panel) and a set of
    node-specific bandwidths, one forms a new graph (shown in center
    panel) by connecting any pair of nodes within
    $B_j$ hops of node $j$ in $G$ (five of the 20 nodes have
    nonzero bandwidths with values shown in figure; the remaining 15
    nodes have zero bandwidths).  Definition \ref{def:banded-vb} refers
    to the resulting adjacency matrix (right panel) as $(B_1,\ldots,B_p)$-banded with
  respect to the graph $G$.}
  \label{fig:ggb-vb}
\end{figure}

We propose two estimators in this paper that make use of a {\em
  latent overlapping group lasso} penalty
\citep{jacob2009group,Obozinski11}, which is used in situations in which one desires a sparsity pattern that is a
union of a set of predefined groups.  Both Definitions \ref{def:banded} and
\ref{def:banded-vb} can be described as such, so the difference
between the two methods lies in the choice of groups. 
In Section \ref{sec:ggb}, we introduce and study the {\em graph-guided
  banding} estimator, which produces estimates of the covariance matrix
that are banded in the sense of Definition \ref{def:banded}.
The limited number of sparsity patterns attainable by this estimator
makes it potentially undersirable in certain situations, and we therefore introduce and study (in Section
\ref{sec:ggb-vb}) the {\em graph-guided banding with local
  bandwidths} estimator, which attains bandedness in the more flexible sense of
Definition \ref{def:banded-vb}.  Section \ref{sec:empirical-study} compares
these methods empirically both in simulation and in naturally
occurring data.  Section \ref{sec:discussion} concludes with a
discussion of related work, insights learned, and future directions.

To our knowledge, there is little work on using known graphs in
high-dimensional covariance estimation.  \citet{cai2015minimax} study minimax lower bounds and
adaptive estimation of covariance matrices when variables lie on a
$d$-dimensional grid.  \citet{deng2009large} consider inverse covariance estimation in the setting where
  variables lie in a metric space.  Their approach uses a non-negative Garrotte semidefinite program formulation to
  perform inverse covariance estimation in which inverse covariance
  decays with increasing distance between variables.  The method
  requires an initial estimate of the inverse covariance matrix (when
  $n>p$, one takes $S^{-1}$) and then shrinks so that variables
  farther apart have smaller partial
  covariances.

\noindent{\bf Notation:} Given a $p\times p$ matrix $\bfA$ and a set
of indices $g\in \{1,\ldots, p\}^2$, we write $\bfA_g$ to denote the $p\times p$
matrix that agrees with $\bfA$ on $g$ and is zero on $g^c$, and we
write $|g|$ for the cardinality of the set $g$.  
The graph $G^B$ is the $B$th power of $G$, meaning
the graph in which any nodes that are within $B$ of each other in $G$
are connected.
  The constraint $\bSigma\succeq\delta I_p$ means that
the smallest eigenvalue of $\bSigma$ is at least $\delta$.  Given a
square matrix $\bSigma$, the matrix $\bSigma^-=\bSigma-\diag(\bSigma)$
agrees with $\bSigma$ on off-diagonals but is zero on the
diagonal. We write $a_n\lesssim b_n$ to mean that there exists a
constant $C>0$ for which $a_n\le Cb_n$ for $n$ sufficiently large.

\section{Graph-guided banding with global bandwidth}
\label{sec:ggb}

Definition \ref{def:banded} introduces the simplest notion of
bandedness with respect to a graph.  In particular, a single bandwidth
$B$ is postulated to apply globally across the entire graph.  This
means that the support of $\bSigma$ is given by one of the following
groups of parameters:
\begin{align}
  g_b=\{jk:1\le d_G(j,k)\le b\}\qquad\text{for }b=1,\ldots, \diam(G),\label{eq:groups-global}
\end{align}
where $\diam(G)$ is the diameter of the seed graph $G$.  The group
lasso \citep{Yuan06} is a well-known extension of the lasso
\citep{Tibshirani96} that is useful for designing estimators in which
groups of variables are simultaneously set to zero.   Given a sample
of $n$ independent random vectors $\X_1,\ldots, \X_n\in\real^p$, form
the sample covariance matrix, $\S=\sum_{i=1}^n(\X_i-\bar
\X_n)(\X_i-\bar \X_n)^T$.  We define the
{\em graph-guided banding estimator with global bandwidth} (\ggbgb) as
\begin{align}
  \hSigma_\lambda=\arg\min_{\bSigma\succeq\delta
    I_p}\left\{\half\|\S-\bSigma\|_F^2+\lambda P(\bSigma;G)\st [\bSigma^-]_{g_M^c}=0 \right\},\label{eq:ggb}
\end{align}
where
$$
P(\bSigma;G)=\min_{\{\V\in\real^{p\times
    p}\}}\left\{\sum_{b=1}^Mw_b\|\V\|_F\st\bSigma_{g_M}=\sum_{b=1}^M\V_{g_b}\right\}
$$
is the latent overlapping group lasso penalty \citep{jacob2009group}
with group structure
given by \eqref{eq:groups-global}.  Each $\V_{g_b}$ matrix has
the ``graph-guided'' sparsity pattern $g_b$, so heuristically, this penalty
favors $\hSigma_\lambda$ that can be written as the sum of a small
number of such sparse matrices.  While the penalty is itself defined by an optimization problem, in
practice we can solve \eqref{eq:ggb} directly without
needing to evaluate $P(\bSigma;G)$, as described in the next section.

The parameter $M$ will typically
be taken to be $\diam(G)$; however, in certain applications in which
$\diam(G)$ is very large and a reasonable upper bound on $B$ is available, it may be computationally expedient to
choose a much smaller $M$.  The constraint $[\bSigma^-]_{g_M^c}=0$
only has an effect when $M < \diam(G)$.  Of course, one must be
cautious not to inadvertently choose $M < B$ since then
the resulting estimates will always be sparser than the true
covariance matrix.  Thus, if one does not have a reliable upper bound
for $B$, one should simply take $M=\diam(G)$.  The tuning parameter $\lambda\ge 0$
controls the sparsity level of $\hSigma_\lambda$ (with large values of
$\lambda$ corresponding to greater sparsity).  The $w_g\ge0$ are
weights that control the relative strength of the individual $\ell_2$
norms in the penalty.  In studying the theoretical properties of the
estimator (Section \ref{sec:theory-ggb}), it becomes clear that
$w_b=|g_b|^{1/2}$ is a good choice.  The parameter $\delta$ is a
user-specified lower bound on the eigenvalues of $\hSigma_\lambda$.
While our algorithm is valid for any value of $\delta$, the most
common choice is $\delta=0$, which ensures that $\hSigma_\lambda$ is a
positive semidefinite matrix.  In certain applications, a positive
definite matrix may be required, in which case the user may specify a
small positive value for
$\delta$.
 
In the special case that $G$ is a path graph (and dropping the
eigenvalue contraint), the \ggbgb~estimator
reduces to the estimator introduced in \citet{yan2015hierarchical}.
In the next section, we 
show that \eqref{eq:ggb} can be efficiently computed.

\subsection{Computation}
\label{sec:computation-ggb}

The challenge of solving \eqref{eq:ggb} lies in the combination of a
nondifferentiable penalty and a positive semidefinite constraint.
Without the penalty term, $\hSigma_\lambda$ could be computed with a
single thresholding of the eigenvalues; without the eigenvalue
constraint, $\hSigma_\lambda$ is nothing more than the proximal
operator of $\lambda P(\cdot;G)$ evaluated at $\S$.
The following lemma shows that we can solve \eqref{eq:ggb} by alternate application of the proximal operator of
 $\lambda P(\cdot;G)$ and a thresholding of the eigenvalues.  We
 present this result for general convex penalties since the same
 result will be used for the estimator introduced in Section \ref{sec:ggb-vb}.
 \begin{lem}
   Consider the problem
   \begin{align}
     \hSigma=\arg\min_{\bSigma\succeq\delta
       I_p}\left\{\half\|\S-\bSigma\|_F^2+\lambda\Omega(\bSigma)\right\},\label{eq:general}
   \end{align}
where $\Omega$ is a closed, proper convex function, and let
$$
\text{Prox}_{\lambda\Omega}({\bf M}) = \arg\min_{\bSigma}\left\{\half\|{\bf M}-\bSigma\|_F^2+\lambda\Omega(\bSigma)\right\}
$$
denote the proximal operator of $\lambda\Omega(\cdot)$ evaluated at a
matrix ${\bf M}$.
Algorithm \ref{alg:bcd-dual}
converges to $\hSigma$.
\label{lem:bcd-for-psd}
 \end{lem}
 \begin{proof}
   See Appendix \ref{supp-sec:bcd-derivation}.
 \end{proof}
\begin{algorithm}
  \caption{\small\em Algorithm for solving \eqref{eq:general}}
  \begin{itemize}
  \item Initialize $\C$.
  \item Repeat until convergence:
    \begin{enumerate}
    \item $\B\leftarrow
      (\S+\C)-\text{Prox}_{\lambda\Omega}(\S+\C)$
    \item $\C\leftarrow\U\cdot\diag([\Lambda_{ii}+\delta]_+)\cdot\U^T$
      where  $\B-\S=\U\Lambda\U^T$ is the eigendecomposition. 
    \end{enumerate}
  \item Return $\hSigma\leftarrow \S-\B+\C$.
  \end{itemize}
\label{alg:bcd-dual}
\end{algorithm}
 In Appendix \ref{supp-sec:bcd-derivation}, we show that Algorithm \ref{alg:bcd-dual}
corresponds to
blockwise coordinate ascent on the dual of \eqref{eq:general}.  While the
{\em alternating direction method of multipliers} (ADMM) would be an
alternative approach to this problem, we prefer Algorithm
\ref{alg:bcd-dual} because, unlike ADMM, its use does not require
choosing a suitable algorithm-specific parameter.
Invoking the lemma, all
that remains for solving \eqref{eq:ggb} is to specify an algorithm for evaluating $\text{Prox}_{\lambda
  P(\cdot;G)}$.  The group structure \eqref{eq:groups-global} used in this estimator
is special in that
$$
g_1\subset g_2\subset\cdots\subset g_M.
$$
\citet{yan2015hierarchical} study the use of the latent overlapping group lasso with
such hierarchically nested group structures and
show that, for group structures of this form, the proximal operator can be evaluated essentially in closed
form (which is not true more generally).  Algorithm 3 of \citet{yan2015hierarchical}
shows
that evaluating the proximal operator amounts to identifying a
sequence of ``breakpoints'' $k_0<k_1<\ldots<k_m$ and then applying
groupwise soft-thresholding to the elements in $g_{k_{i+1}}-g_{k_i}$
for $i=0,\ldots,m-1$.

In practice, we solve \eqref{eq:general} for a decreasing sequence of
values of $\lambda$.  We initialize $\C$ in Algorithm
\ref{alg:bcd-dual} as ${\bf 0}$ for the first
value of $\lambda$ and then for each subsequent value of $\lambda$ we use
the value of $\C$ from the previous
$\lambda$ as a warm start. 

\subsection{Theory}
\label{sec:theory-ggb}

In this section, we study the statistical properties of the GGB
estimator \eqref{eq:ggb}.  For simplicity, we will focus on the
special case in which $\delta=-\infty$ (or, equivalently, we
drop the eigenvalue constraint); however, in Section \ref{sec:eigen}
we provide conditions under which our results apply to the $\delta=0$ case.  We introduce two conditions
under which we prove these results.
\begin{ass}[distribution of data]
We draw $n$ i.i.d. copies of $\X\in\real^p$, having mean ${\bf 0}$ and
covariance $\bSigma^*$.  The random vector has bounded variances,
$\max_j|\bSigma^*_{jj}|\le \bar\kappa$, and sub-Gaussian marginals: $\mathbb
E[e^{t\X_{j}/\sqrt{\bSigma^*_{jj}}}]\le e^{Ct^2}$ for all $t\ge0$.
Both $C$ and $\bar\kappa$ are constants.
\label{ass:distribution}
\end{ass}
\begin{ass}[relation between $n$ and $p$]
We require that $p\le e^{\gamma_1 n}$ for some $\gamma_1>0$.  
\label{ass:scaling}
\end{ass}
\noindent Assumptions \ref{ass:distribution} and \ref{ass:scaling} are
sufficient for establishing that
the sample covariance matrix $\S$ concentrates around $\bSigma^*$ with high probability for large $n$.  These
assumptions are made in \citet{Bien15convex}, and similar assumptions appear throughout the literature.

\subsubsection{Frobenius norm error}
\label{sec:frobenius-norm-error}

We begin by producing the rate of convergence of our estimator in
Frobenius norm.
  \begin{thm}
  Suppose $\bSigma^*$ is $B^*$-banded with respect to the graph $G$
  (in the sense of Definition \ref{def:banded}) and Assumptions
  \ref{ass:distribution}--\ref{ass:scaling} hold.
  Then with $\lambda=x\sqrt{\frac{\log(\max\{p,n\})}{n}}$ (for $x$
  sufficiently large), $w_b=|g_b|^{1/2}$,
  $\delta=-\infty$, and $M=\diam(G)$, there exists a constant $c>0$
  such that
$$
\|\hat\bSigma_\lambda - \bSigma^*\|_F^2\lesssim \frac{(p+|g_{B^*}|)\log(\max\{p,n\})}{n}
$$
with probability at least $1-c/\max\{p,n\}$.
\label{thm:ggb}
\end{thm}
\begin{proof}
  See Appendix \ref{supp-app:proof-ggb}.
\end{proof}
The term $p+|g_{B^*}|$ is the sparsity level of $\bSigma^*$ (with the
initial $p$ representing the diagonal elements).  Theorem
\ref{thm:ggb} shows that when $p<n$ and moreover $p+|g_{B^*}|$
is $o(n/\log n)$, the GGB estimator is consistent in Frobenius norm.  By
contrast, consistency of the sample covariance matrix $\S$ in this
setting requires that $p^2$ be $o(n/\log n)$.  Thus, when $|g_{B^*}|\ll p^2$, GGB is expected to
substantially outperform the sample covariance matrix.  When $p>n$, the theorem
does not establish consistency of GGB.  This fact is not, however, a shortcoming of
the GGB method; rather, it is a reflection of the difficulty of the
problem and the fact that the assumed sparsity
level is in this case high relative to $n$.  Even when $B^*=0$ (i.e.,
$\bSigma^*$ is a diagonal matrix), we still must estimate $p$ free
parameters on the basis of only $n$ observations.  More generally,
we know that the rate of convergence in Theorem
\ref{thm:ggb} cannot, up to log factors, be improved by any other
method.  To see this, note that in the special case that $G$ is a
path graph, \citet{Bien15convex} prove (in Theorem 7 of their paper) that the optimal rate is
$B^*p/n$ over the class of $B^*$-banded matrices (assuming
$B^*\le\sqrt{n}$).  Since $|g_{B^*}|\sim B^*p$ for a path graph, we
see that the GGB estimator is within a log factor of the optimal rate
in Frobenius norm.

If we are interested in estimating the correlation matrix $\corr=\diag(\bSigma^*)^{-1/2}\bSigma^*\diag(\bSigma^*)^{-1/2}$ instead of
the covariance matrix, then we can apply our estimator to the sample
correlation matrix
$\hcorr=\diag(\S)^{-1/2}\cdot\S\cdot\diag(\S)^{-1/2}$.  Because the
diagonal entries of the correlation matrix do not need to be
estimated, we can show (under one additional assumption) that the rate of
convergence has $|g_{B^*}|$ in place of $p+|g_{B^*}|$.

\begin{ass}[bounded non-zero variances]
Suppose that there exists a constant $\kappa>0$ (i.e., not depending on $n$ or $p$) such that
$\min_j\bSigma^*_{jj}\ge \kappa$.

\label{ass:bounded}
\end{ass}

  \begin{thm}
  Suppose $\bSigma^*$ is $B^*$-banded with respect to the graph $G$
  (in the sense of Definition \ref{def:banded}) and Assumptions
  \ref{ass:distribution}, \ref{ass:scaling}, and \ref{ass:bounded} hold.
  Then there exist positive constants $x$ and $c$ for which taking $\lambda=x\sqrt{\frac{\log(\max\{p,n\})}{n}}$, $w_b=|g_b|^{1/2}$,
  $\delta=-\infty$, and $M=\diam(G)$, 
$$
\|\hat\bSigma_\lambda(\hcorr) - \corr\|_F^2\lesssim \frac{|g_{B^*}|\log(\max\{p,n\})}{n}
$$
with probability at least $1-c/\max\{p,n\}$.
\label{thm:ggb-corr}
\end{thm}

When $G$ has a small number of edges (e.g., with all but a small
number of nodes having 0 degree), then we can have $|g_{B^*}|\ll n <
p$, in which case the above result establishes the consistency of the
graph-guided estimator even when $p>n$. 

\subsubsection{Bandwidth recovery}
\label{sec:bandwidth-recovery}

The remaining properties of this section can be established under a
condition on the minimum signal strength (analogous to the
``beta-min'' conditions arising in the regression context).
\begin{ass}[minimum signal strength]
The nonzero elements of $\bSigma^*$ must be sufficiently large:
 $$
\min_{jk\in g_{B^*}}|\bSigma^*_{jk}| > 2\lambda.
$$
\label{ass:signal-strength}
\end{ass}
One can think of $\lambda$ as on the order of
$\sqrt{\frac{\log(\max\{p,n\})}{n}}$ and measuring the ``noise level''
of the data, in particular $\|\S-\bSigma^*\|_\infty$. Thus, this
assumption ensures that the nonzeros stand out from the noise.  Under
this additional assumption, our estimator's bandwidth, $B(\hSigma_\lambda)$, matches the unknown
bandwidth $B^*$ with high probability.
\begin{thm}
Adding Assumption \ref{ass:signal-strength} to the conditions of
Theorem \ref{thm:ggb}, there exists a constant $c>0$ such that
$
B(\hSigma_\lambda) = B^*
$
with probability at least $1-c/\max\{p,n\}$.
\label{thm:bandwidth-recovery}
\end{thm}
\begin{proof}
  See Appendix \ref{supp-app:bandwidth-recovery}.
\end{proof}

\subsubsection{Positive semidefinite constraint}
\label{sec:eigen}

The results presented above are for $\delta=-\infty$.  The next result provides
an additional condition under which the previous results hold for
the $\delta=0$ case (i.e., the situation in which the estimator is
forced to be positive semidefinite).

Suppose all nonzero elements of $\bSigma^*$ are within a factor of
$\tau$ of each other:
$$
\frac{\max_{jk\in g_{B^*}}|\bSigma^*_{jk}|}{\min_{jk\in
    g_{B^*}}|\bSigma^*_{jk}|}\le \tau.
$$

\begin{ass}[minimum eigenvalue]
The minimum eigenvalue of $\bSigma^*$ must be sufficiently large:
 $$
\lambda_{\min}(\bSigma^*)> 4x\tau\sqrt{\log(p)/n}\max_j|g_{jB^*}|.$$
\label{ass:eigenvalue}
\end{ass}
The quantity $\max_j|g_{jB^*}|$ is the maximal degree of the true
covariance graph (defined by the sparsity pattern of $\bSigma^*$).
\begin{thm}
Adding Assumptions \ref{ass:signal-strength} and \ref{ass:eigenvalue} to the conditions of Theorem \ref{thm:ggb},
Theorems \ref{thm:ggb} and \ref{thm:bandwidth-recovery} hold for $\delta=0$.
\label{thm:op}
\end{thm}

\subsection{Discussion \& the small world phenomenon}
\label{sec:discussion}

The theorems above assume that $M=\diam(G)$, but the results in fact
apply more generally so long as $M\ge B^*$.  While in general we may not
know $B^*$, there are some applications in which one may know enough
about the mechanism of the graph-based dependence to provide an upper
bound on $B^*$.  In such situations, choosing an $M$ that is much
smaller than $\diam(G)$ may lead to substantial computational
improvements (both in terms of timing and memory).

Examining the Frobenius-norm 
bound above, we see that the
structure of the seed graph $G$ plays an important role in the rates
of convergence.  While it is tempting to think of $B^*$ as controlling
the sparsity level (and therefore the difficulty of the problem), it is
important to observe that this quantity is in fact a function of both $B^*$ and the graph $G$.
For example, in the Frobenius-norm bound, the quantity $|g_{B^*}|$ is
the number of pairs of variables that are within $B^*$ hops of each other in the
graph $G$.  In conventional banded estimation, where $G$ is a path
seed graph, the dependence on $B^*$ is linear: $|g_{B^*}|\sim B^*p$.
More generally, when $G$ is a $d$-dimensional lattice graph
$|g_{B^*}|\sim (B^*)^dp$ (when $B^*$ is small enough to avoid edge
effects).  However, intuition provided by lattice graphs does not
extend well to more general graphs.  For a particularly extreme example, when $G$ is a complete
graph, then $|g_1|={p \choose 2}$.  Many large graphs in nature have small
diameter \citep{watts1998collective}, which means that $|g_{B^*}|$ may
be quite large even when $B^*$ is small.  The idea of ``six degrees of
separation'' \citep{travers1969experimental} means that when using a
social network as a seed graph, $|g_6|\sim p^2$ (in fact, recent work
suggests six might even be an overestimate,
\citealt{backstrom2011four}).  

When $\diam(G)$ is very small, Definition \ref{def:banded} is
restricted to a very limited class of sparsity patterns.  In the next
section, we consider a much larger class of sparsity patterns, using
the notion of local bandwidths introduced in Definition \ref{def:banded-vb}.

\section{Graph-guided banding with local bandwidths}
\label{sec:ggb-vb}
Definition \ref{def:banded-vb} describes a different kind of
graph-guided banding, in which each variable $j$ has an associated
bandwidth $B_j^*$, which can be thought of as a neighborhood size.
Consider sets of the form
$$
g_{jb}=\{k:1\le d_G(j,k)\le b\}, \text{ for } 1\le b\le M_j,\quad 1\le j\le p,
$$
which is the set of variables within $b$ hops of $j$ in the seed graph
$G$.  We can describe a sparsity pattern of the type in Definition
\ref{def:banded-vb} as
$$  S^*=\bigcup_{1\le j\le p}\left[\{j\}\times g_{jB_j^*}\right]\cup \left[g_{jB_j^*}\times \{j\}\right].$$
To see this, note that $jk\in S^*$ is equivalent to $d_G(j,k)\le \max\{B_j^*,B_k^*\}$. 
We desire an estimator whose sparsity patterns are unions of the
$\{g_{jb}\}$, so we use the latent overlapping group lasso penalty
\citep{jacob2009group}, which is designed for such situations.  We define the
{\em graph-guided banding estimator with local bandwidths} (\ggbvb)
as 
\begin{align}
  \tSigma_\lambda=\arg\min_{\bSigma\succeq\delta
    I_p}\left\{\half\|\S-\bSigma\|_F^2+\lambda \tilde P(\bSigma;G) \right\},\label{eq:ggb-vb}
\end{align}
where
$$
\tilde P(\bSigma;G)=\min_{\{\V[jb]\in\real^{p\times
    p}\}}\left\{\sum_{j=1}^p\sum_{b=1}^{M_j}w_{jb}\|\V[jb]\|_F\st\bSigma_{g_{jM_j}}=\sum_{b=1}^{M_j}\V[jb]_{g_{jb}}\right\}.
$$
The only difference between this estimator and the \ggbgb~estimator of
\eqref{eq:ggb} is in the choice of groups.  For a fixed $j$, the
groups are nested, $g_{j1}\subseteq\cdots\subseteq g_{j M_j}$.  When
$M_j$ is the diameter of the connected component that contains $j$,
then 
$g_{j M_j}$ represents all variables that are connected to variable
$j$.  For $j\neq k$, we can have $g_{jb}\cap g_{kb}\neq\emptyset$ without one group contained completely in the other.  The fact that this
group structure is not simply
hierarchically nested, as in the global bandwidth estimator, complicates
both the computation and the theoretical analysis of the estimator.

\subsection{Computation}
\label{sec:computation-ggb-vb}

By Lemma \ref{lem:bcd-for-psd}, we can solve \eqref{eq:ggb-vb} by
Algorithm \ref{alg:bcd-dual}.  The challenge of evaluating the
\ggbvb~estimator lies therefore in efficiently evaluating the proximal
operator $\text{Prox}_{\lambda\tilde P}$.  Viewing this problem as
 an optimization problem over the latent variables $\sV[jb]$ suggests a
 simple blockwise coordinate descent approach, which we adopt here.  

\subsection{Theory}
\label{sec:theory-ggb-vb}

In Theorem \ref{thm:ggb}, we saw that the Frobenius-norm convergence
depends on the number of nonzero off-diagonal elements, $|g_{B*}|$, a
quantity that we observed in Section
\ref{sec:discussion} could be quite large---as large as $O(p^2)$---even when $B^*$ is very small.  In other words, for such
``small-world'' seed graphs $G$, the assumption of graph-guided
bandedness with a small global bandwidth does not necessarily correspond to
$\bSigma^*$ being a sparse matrix.  The notion of graph-guided
bandedness with {\em local} bandwidths, as in Definition \ref{def:banded-vb},
provides finer control, allowing us to describe sparse $\bSigma^*$
even with ``small-world'' seed graphs $G$.  The purpose of the next theorem
is to show that, under such a notion of sparsity, the \ggbvb~estimator
can attain Frobenius-norm consistency when $\bSigma^*$ is sparse. 
\begin{thm}
Suppose $\bSigma^*$ is $(B_1^*,\ldots,B_p^*)$-banded with respect to a
graph $G$ (in the sense of Definition \ref{def:banded-vb}) and
Assumptions \ref{ass:distribution}-\ref{ass:scaling} hold. 
Then with $\lambda=x\sqrt{\frac{\log(\max\{p,n\})}{n}}$ (for $x$
  sufficiently large), $w_{jb}=|g_{jb}|^{1/2}$,
  $\delta=-\infty$, and $M_j=\max_k d_G(j,k)$,
$$
\|\tSigma_\lambda-\bSigma^*\|_F^2\lesssim|S^*|\sqrt{\frac{\log(\max\{p,n\})}{n}}+\frac{p\log(\max\{p,n\})}{n}
$$
 with probability at least $1-c/\max\{p,n\}$.
\label{thm:slowrate}
\end{thm}
\begin{proof}
  See Appendix \ref{supp-sec:slowrate}.
\end{proof}
\noindent The two terms on the right-hand side of Theorem \ref{thm:slowrate}
represent estimation of the off-diagonal and diagonal elements of
$\bSigma^*$.  If $\bSigma^*$ is sparse enough---in particular that
$|S^*|\le p \sqrt{\frac{\log(\max\{p,n\})}{n}}$, then estimating all
of $\bSigma^*$ is as difficult as estimating its diagonal elements
alone (i.e., just the $p$ variances).  A simple example of this
situation is when
$\bSigma^*$ is  $(B_1^*,\ldots,B_p^*)$-banded with respect to a
graph $G$ and $B_j^*>0$ for at most 
$\sqrt{\frac{\log(\max\{p,n\})}{n}}$ variables $j$.  Another example
is when all $B_j^*\le1$ and the sum of degrees of all $j$ with
$B_j^*=1$ is bounded by $p\sqrt{\frac{\log(\max\{p,n\})}{n}}$.

Theorem \ref{thm:slowrate} suggests that consistency in Frobenius norm
might not be achieved in the $p>n$ regime even in the sparsest situation of a
diagonal matrix.  As noted earlier, this is not a shortcoming of our
method but rather a reflection of the
difficulty of the task, which requires estimating $p$ independent
parameters even in the sparsest situation of a diagonal matrix.  

As noted previously, if instead we were interested in estimating the
correlation matrix, we can attain a rate as in Theorem \ref{thm:slowrate}
but without the second term by applying our estimator to the sample
correlation matrix.
\begin{thm}
Suppose $\bSigma^*$ is $(B_1^*,\ldots,B_p^*)$-banded with respect to a
graph $G$ (in the sense of Definition \ref{def:banded-vb}) and
Assumptions \ref{ass:distribution}, \ref{ass:scaling}, and \ref{ass:bounded} hold. 
Then there exists a constant $x$ for which taking $\lambda=x\sqrt{\frac{\log(\max\{p,n\})}{n}}$, $w_{jb}=|g_{jb}|^{1/2}$,
  $\delta=-\infty$, and $M_j=\max_k d_G(j,k)$,
$$
\|\tSigma_\lambda(\hcorr)-\corr\|_F^2\lesssim|S^*|\sqrt{\frac{\log(\max\{p,n\})}{n}}
$$
 with probability at least $1-c/\max\{p,n\}$.
\label{thm:slowrate-corr}
\end{thm}
Consistency follows even when $p\gg n$ as long as the right-hand side
approaches 0.

\section{Empirical study}
\label{sec:empirical-study}

\subsection{Simulations}
\label{sec:simulations}

 All simulations presented in this paper were performed using the {\tt
 simulator} package \citep{bien2016simulator,simulator} in {\tt R} \citep{R}. 
 We investigate the advantages of exploiting a known
graph structure (both in the global and local bandwidths settings),
and we also study our two methods' robustness to graph misspecification.
We compare the two GGB methods to each other and to two other
methods, which represent the ``patternless sparsity'' and traditional
banding approaches.

Perhaps the simplest and most commonly used sparse covariance
estimator is to simply soft-threshold each off-diagonal element of the
sample covariance matrix:
$$
\hSigma^{\text{soft}}_{jk}=\text{sign}(\S_{jk})\left[|\S_{jk}|-\lambda\right]_+.
$$
This method \citep{rothman2009generalized}, along with its close relatives \citep{Rothman12,Liu13,xue2012positive}, do not make
use of the graph $G$ and corresponds to \eqref{eq:general} with
$\Omega$ being the $\ell_1$ norm and $\delta=-\infty$.
Thus, comparisons to this method allow us to
understand what is gained and what is lost by using the graph $G$
versus seeking ``patternless sparsity.''

We also wish to compare our new estimators to a more traditional
(i.e., non-graph-guided) banding approach.
Graph-guided banding is not the only way one might incorporate the graph
$G$.  A different approach would be to find an ordering of the nodes
based on $G$ and then apply a traditional banding estimator with respect
to this ordering.  To this end, we seek an ordering of the variables that travels
along $G$ as much as it can, making jumps to another
part of $G$ as sparingly as possible.  We express this as a traveling salesman problem (TSP),
which we then solve using an off-the-shelf TSP solver (see Appendix
\ref{supp-sec:tsp} for details and for images of the learned paths).
We then apply the {\em convex banding} method using this ordering \citep{Bien15convex,hierband}.

While each of these methods can be solved with a positive semidefinite
constraint, we compare their performances
without such a constraint.  Doing so allows us to observe the
fraction of the time each unconstrained estimator is positive semidefinite: Table \ref{supp-tab:psd} in Appendix
\ref{supp-sec:psd} shows that the banding estimators are observed in
the study to always be positive semidefinite, whereas this is not the
case for the soft-thresholding estimator.  It follows that the
reported results for the banding methods in fact also pertain to these
estimators with $\delta=0$.

To generate a covariance matrix that is $(B_1^*,\ldots,B_p^*)$ banded with
respect to a graph $G$, we take
$$
\bSigma^*_{jk}=
\begin{cases}
 \frac1{d_G(j,k)}1\left\{d_G(j,k)\le\max(B^*_j,B^*_k)\right\}&\text{ if }j\neq k\\
a&\text{ if }j=k,
\end{cases}
$$
where $a$ is chosen to ensure that the minimum eigenvalue of $\bSigma$
is at least $\sigma^2$ (we take $\sigma=0.01$ throughout).

Initially, we compare the methods under four GGB structures:
two-dimensional lattice with global bandwidth,
two-dimensional lattice with local bandwidths,
scale-free graph with global bandwidth, and
scale-free graph with local bandwidths.
In the first two scenarios, we take $G$ to be a two-dimensional square
lattice having 20 variables per side (so $p=400$).  The global
bandwidth is taken to be $B^*=4$, and the local bandwidths are
generated according to 
$$
B^*_j=\begin{cases}
  0&\text{ with probability } 0.90\\
  1&\text{ with probability } 0.06\\
  4&\text{ with probability } 0.04
\end{cases}
$$
(see Figure \ref{fig:lattice}). 
\begin{figure}
  \centering
  \includegraphics[width=0.25\linewidth]{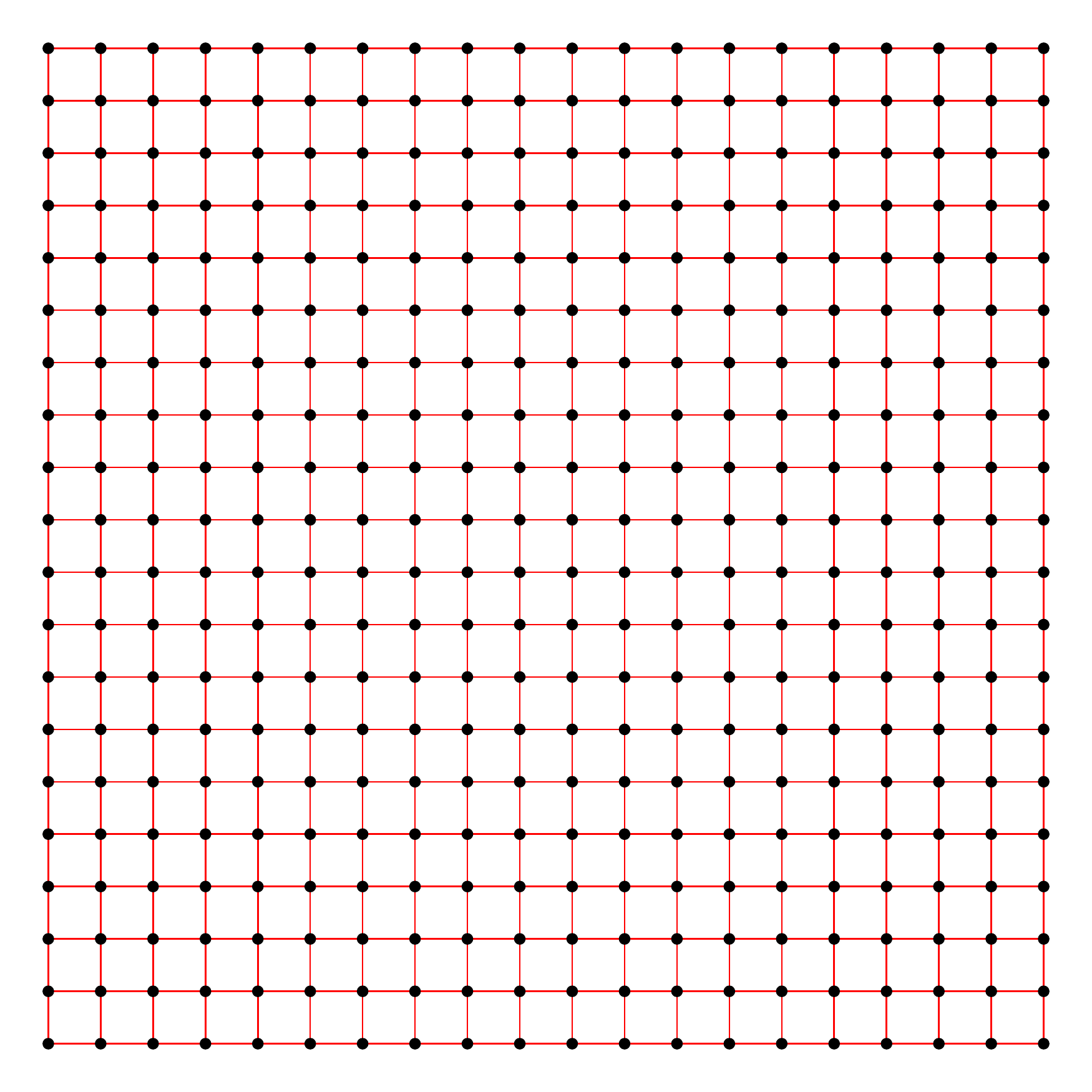}
  \includegraphics[width=0.25\linewidth]{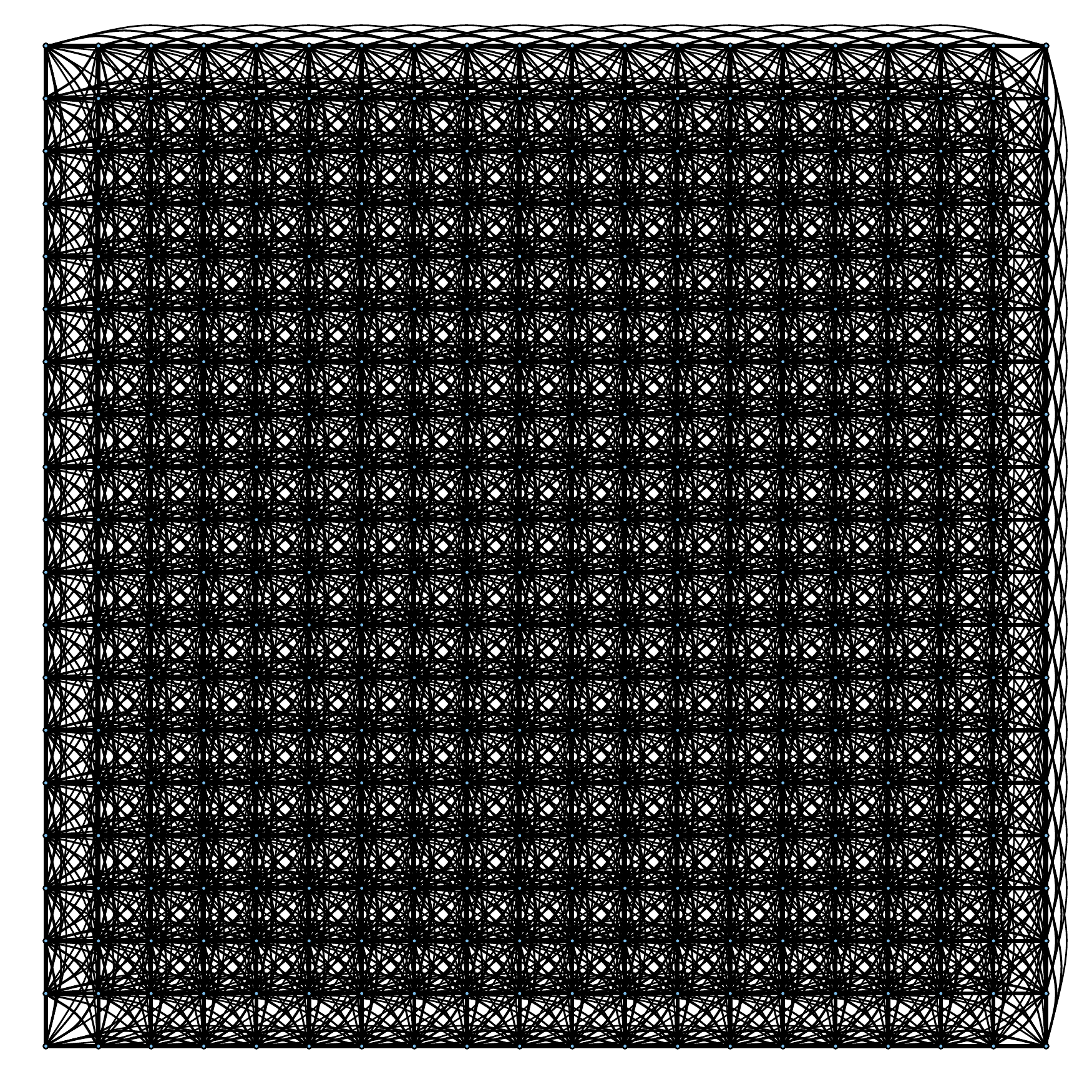}
  \includegraphics[width=0.25\linewidth]{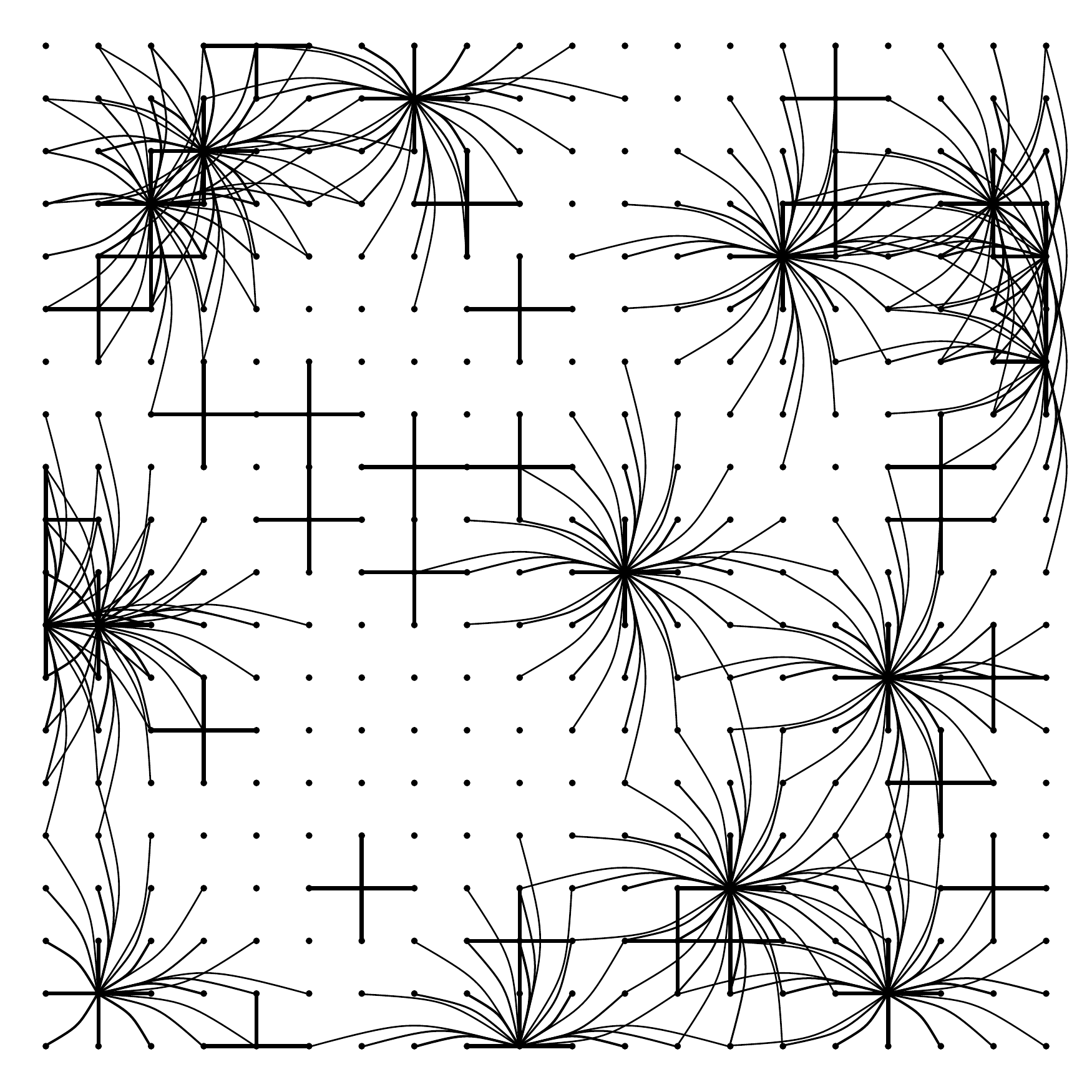}
  \caption{\small\em Two-dimensional lattice graph; seed graph (left
    panel), global bandwidth (center panel), local bandwidths
    (right panel)}
  \label{fig:lattice}
\end{figure}
\begin{figure}
  \centering
  \includegraphics[width=0.25\linewidth]{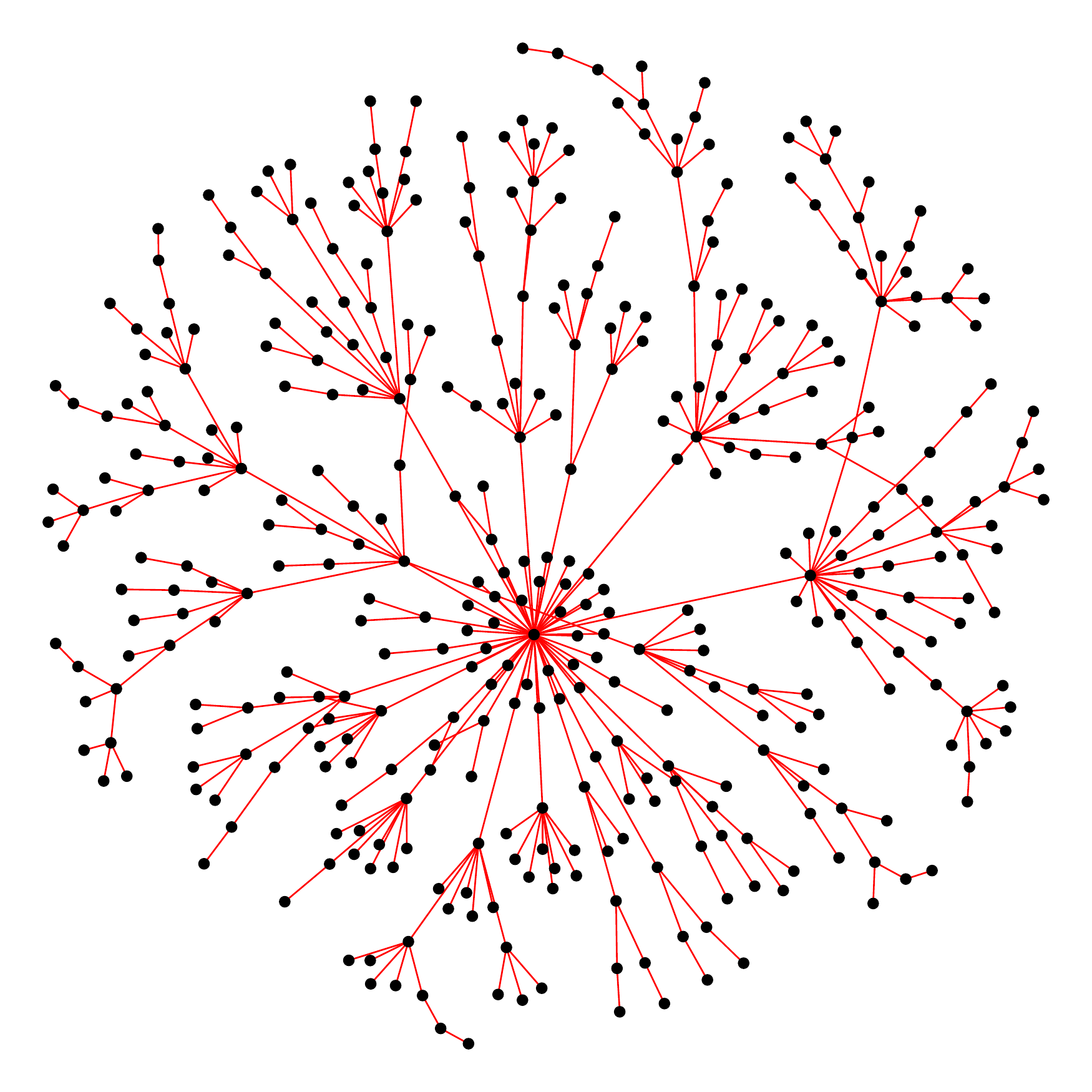}
  \includegraphics[width=0.25\linewidth]{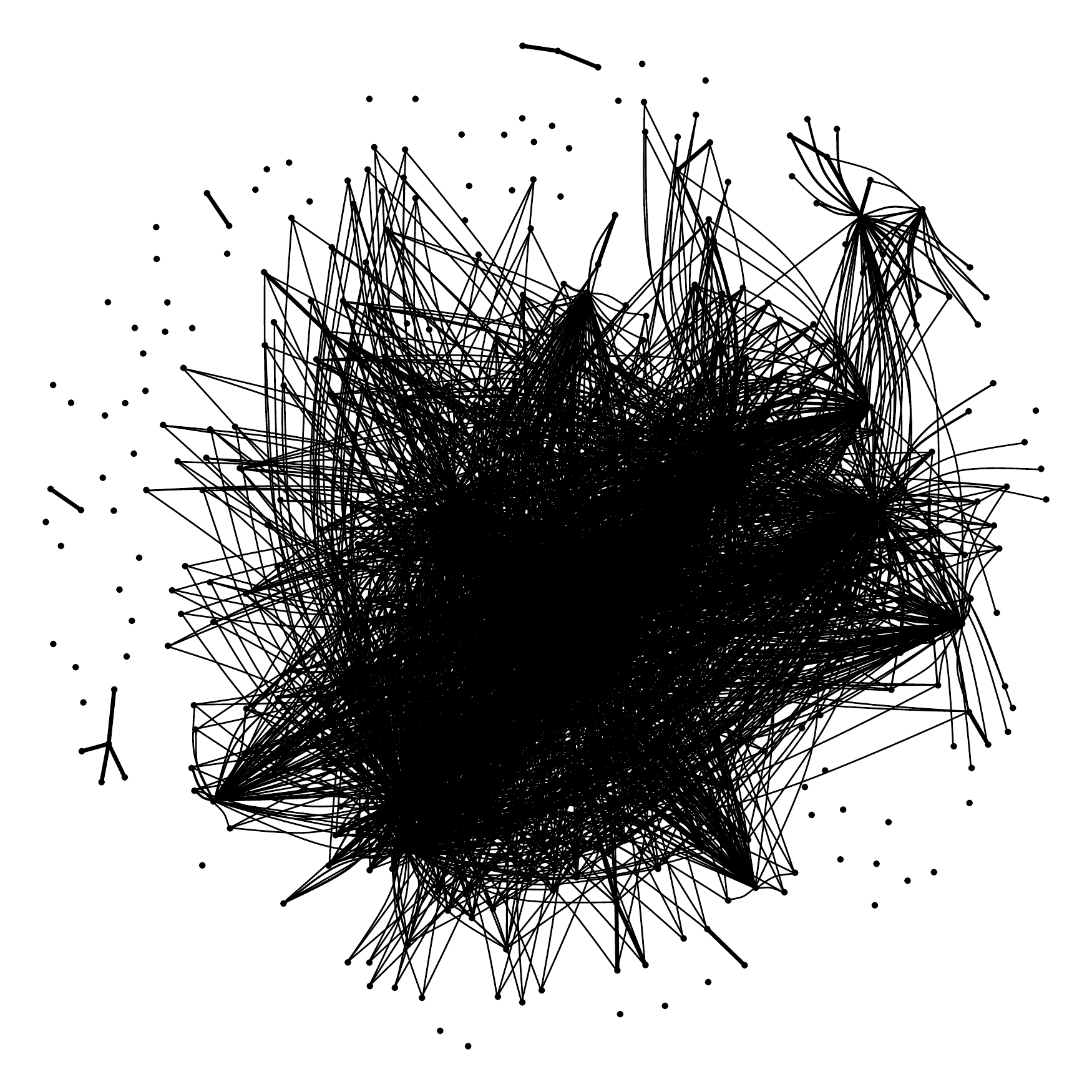}
  \caption{\small\em Scale-free graph; seed graph (left panel), 
local bandwidths (right panel)}
  \label{fig:scale-free}
\end{figure}
The resulting $\bSigma^*$ has roughly 10\%
of elements nonzero in the global case and about 1\% in the local
case.  In the third and fourth scenarios, we take $G$ to be a scale
free network (as generated by the function {\tt sample\_pa} in the
{\tt R}
package {\tt igraph}, \citealt{igraph}) with
$p=400$ nodes.  The same two choices of global and local bandwidths
lead to $\bSigma^*$ having, respectively, about 75\% and 10\% nonzeros
(see Figure \ref{fig:scale-free}).  The marked difference in sparsity
levels is a manifestation of the differing neighborhood structures
between these two graph types.

\subsubsection{Lattice versus scale-free graphs and global versus local bandwidths}
\label{sec:lattice-versus-scale}

We generate samples of size $n=300$ under the four scenarios (and
repeat 200 times).  Tables \ref{tab:cvfrob} and \ref{tab:cvop}
present Frobenius-norm error and operator-norm errors, $\|\hSigma-\bSigma^*\|_F$
and $\|\hSigma-\bSigma^*\|_{op}$, respectively; table \ref{tab:cventropy} presents
the entropy error, $-\log\det(\hSigma\bSigma^{*-1})+\trace(\hSigma\bSigma^{*-1})-p$.  Each method has a single tuning
parameter, and we report the performance of the method
with the value of the tuning parameter chosen by 5-fold cross validation.  In
particular, we randomly partition the training data into 5 folds and, for an estimator $\hSigma_\lambda$, choose
$\hat\lambda_{\text{CV-best}}=\arg\min_\lambda\frac1{5}\sum_{k=1}^5\|\S^{(k)}-\hSigma_\lambda(\S^{(-k)})
\|_F$, where $\S^{(k)}$ is the sample covariance matrix in
which only the $k$th fold is included and $\S^{(-k)}$ is the sample covariance matrix in
which the $k$th fold is omitted.  In the 2-d lattice model, \ggbgb~does best when the true model has global bandwidth,
as we would expect it to.  Interestingly, \ggbvb~does best not just
when the true model has local bandwidths but also in the scale-free
model with a global bandwidth.  This is a reminder of the important role that
the seed graph's structure plays in the statistical properties of the estimators.  Finally, Figure \ref{fig:roc}
shows ROC curves for the four methods. The results are consistent with
the results from the tables.  In particular, \ggbgb~does best
in the first scenario, and \ggbvb~does best in the other scenarios.
Interestingly, {\tt TSP + hiernet} does not appear to be an effective
approach.  Unlike the graph-guided approaches, {\tt TSP + hiernet}
does worse than soft-thresholding in nearly every scenario.  For
example, in Figure
\ref{fig:roc}, we see that  {\tt TSP + hiernet}'s ability to identify
nonzero elements of $\bSigma^*$ is no better than random guessing.
This poor performance is to be expected in this scenario since such a
network cannot be well captured by an ordering of the nodes.  For the
two-dimensional lattice graph scenarios, we find that such an approach
is at least partially successful, which makes sense since the ordering
captures some aspects of the underlying seed graph.

\begin{table}

\caption{\label{tab:cvfrob}A comparison of mean Frobenius error at CV-best (averaged over 200 replicates).}
\centering
\scriptsize
\begin{tabular}[t]{l|l|l|l|l}
\hline
  & \ggbgb & \ggbvb & Soft thresholding & TSP + hierband \\
\hline
2-d lattice (20 by 20), b = 4 & {\bf  7.605} (0.015) & 8.381 (0.031) & 11.634 (0.014) & 11.562 (0.031)\\
\hline
2-d lattice (20 by 20) w/ var-bw & 4.244 (0.002) & {\bf  3.170} (0.008) & 4.347 (0.007) & 4.779 (0.003)\\
\hline
Scale free (p = 400), b = 4 & 6.740 (0.009) & {\bf  6.543} (0.015) & 7.655 (0.004) & 7.727 (0.009)\\
\hline
Scale free (p = 400) w/ var-bw & 3.273 (0.002) & {\bf  2.867} (0.007) & 3.552 (0.005) & 3.661 (0.002)\\
\hline
\end{tabular}
\end{table}

\begin{table}

\caption{\label{tab:cvop}A comparison of mean operator error at CV-best (averaged over 200 replicates).}
\scriptsize
\centering
\begin{tabular}[t]{l|l|l|l|l}
\hline
  & \ggbgb & \ggbvb & Soft thresholding & TSP + hierband\\
\hline
2-d lattice (20 by 20), b = 4 & {\bf 2.185} (0.007) & 2.428 (0.015) & 3.224 (0.010) & 3.149 (0.016)\\
\hline
2-d lattice (20 by 20) w/ var-bw & 0.989 (0.003) & {\bf 0.756} (0.005) & 0.997 (0.004) & 1.024 (0.003)\\
\hline
Scale free (p = 400), b = 4 & 4.350 (0.010) & {\bf 4.110} (0.016) & 4.695 (0.004) & 4.709 (0.009)\\
\hline
Scale free (p = 400) w/ var-bw & 1.160 (0.001) & {\bf 1.008} (0.005) & 1.212 (0.002) & 1.238 (0.001)\\
\hline
\end{tabular}
\end{table}

\begin{table}

\caption{\label{tab:cventropy}A comparison of mean entropy error at CV-best (averaged over 200 replicates).}
\scriptsize
\centering
\begin{tabular}[t]{l|l|l|l|l}
\hline
  & \ggbgb & \ggbvb & Soft thresholding & TSP + hierband \\
\hline
2-d lattice (20 by 20), b = 4 & {\bf  7719} ( 24) & 7819 ( 69) & 10234 ( 34) & 11870 ( 73)\\
\hline
2-d lattice (20 by 20) w/ var-bw & 25823 ( 55) & {\bf 18088} (119) & 25358 ( 89) & 27188 ( 62)\\
\hline
Scale free (p = 400), b = 4 & 72606 (146) & {\bf 68087} (191) & 68654 ( 62) & 68518 (119)\\
\hline
Scale free (p = 400) w/ var-bw & 52047 ( 70) & {\bf 44188} (186) & 53332 ( 86) & 54304 ( 76)\\
\hline
\end{tabular}
\end{table}

\begin{figure}
  \centering
  \includegraphics[width=0.75\linewidth]{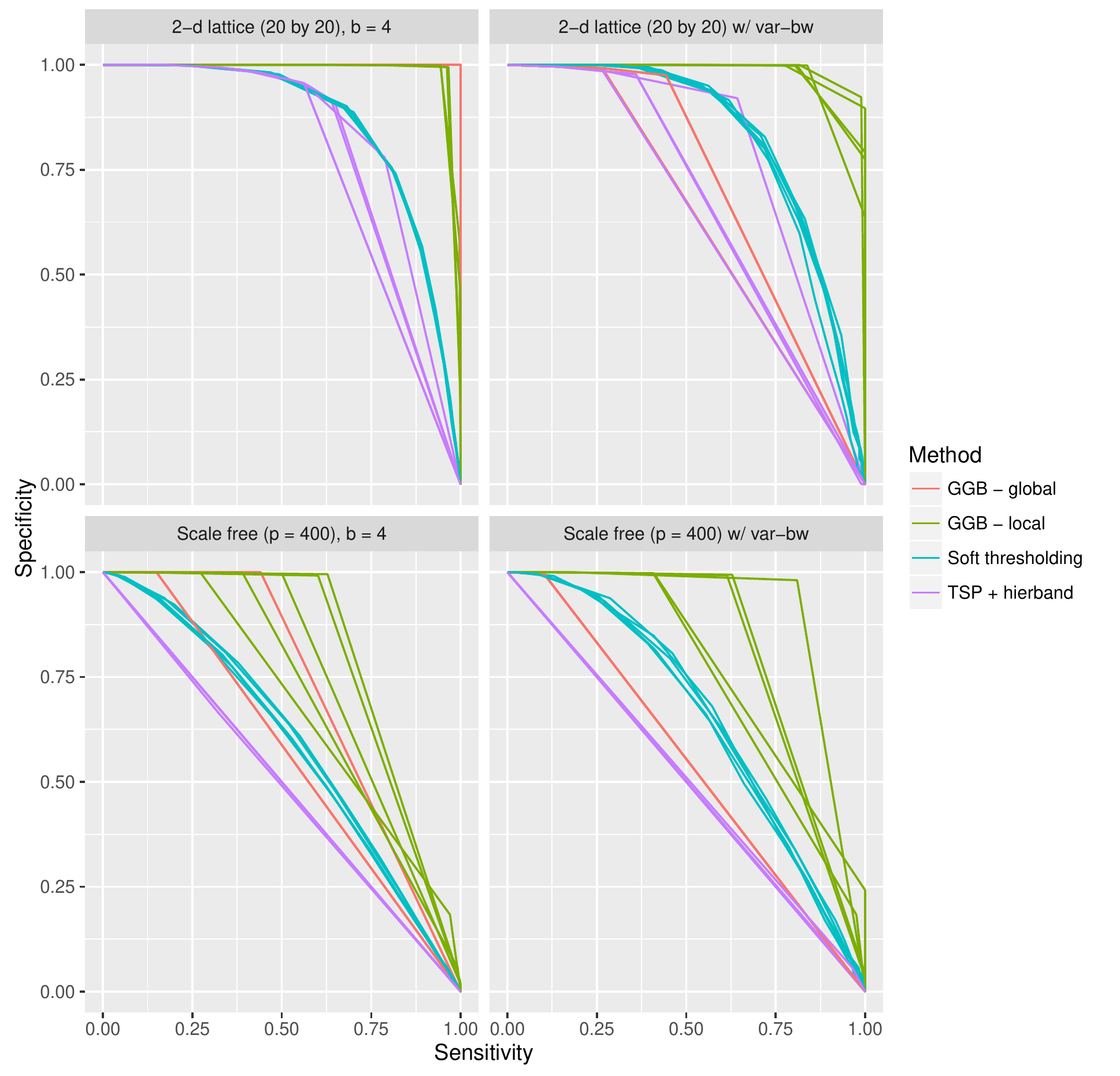}
  \caption{\em ROC curves (formed by varying the tuning parameter for each
    method), show the sensitivity and specificity of identifying the
    nonzero elements of $\bSigma^*$.  For visibility, only five
    of the 200 replicates are shown. (Plot made using {\tt R} package {\tt ggplot2}, \citealt{ggplot2}.)}
  \label{fig:roc}
\end{figure}

\subsubsection{Effect of graph misspecification}
\label{sec:effect-graph-missp}

The previous simulation studies establish the rather obvious fact that when one
knows the true graph $G$ with respect to which $\bSigma^*$ is banded,
one can better estimate $\bSigma^*$ with an estimator that uses $G$.
In practice, however, one might typically use a graph $G'$ that is
(hopefully) close to $G$ but not exactly identical.  This section
seeks to understand the effect of graph misspecification on the
effectiveness of graph-guided banding.  We simulate data as above but
imagine that we only have access to a corrupted version of the
underlying graph $G$.

Let $\pi\in[0,1]$ control the degree of corruption to $G$.  Of course
there are many ways to corrupt a graph, but we choose a scheme that
will preserve the properties of the graph and the sparsity level of the covariance matrix (so as not
to confound graph distortion with other factors such as sparsity
level, number of triangles, etc.).  In particular, we take $G$ to be
an Erdos-Renyi graph with $p=400$ vertices and $760$
edges (this is chosen to match the number of edges in the previous
examples).  We take $\bSigma^*$ to have global graph-guided bandwidth
$B^*=4$ with respect to this graph $G$.
For each edge in $G$, with probability $\pi$ we remove it and then connect two
currently unconnected nodes chosen at random.  We choose $G$ to be an
Erdos-Renyi graph so that the distorted graph $G'$ is itself Erdos-Renyi with the
same number of edges.  By keeping all graph properties fixed, we can be
confident that as we vary $\pi$, any differences in performance of the
methods will be attributable to $G$ and $G'$ having different edges.

Figure \ref{fig:distortion} shows the performance of the three methods
as the probability of reassigning an edge is varied from 0 (no distortion of $G$) to 1 (all edges are
distorted).  As expected, we see both GGB methods deteriorating as
the distortion level of the seed graph $G'$ is increased.  When
$\pi=0$ (i.e., $G=G'$) $\ggbgb$ is best (which makes sense since
$\bSigma^*$ has a global bandwidth).  When approximately $\pi=0.25$ of
the edges are distorted, the best obtainable operator-norm error by
the GGB methods becomes worse than that of soft-thresholding.
Soft-thresholding is, of course, unaffected by graph distortion since
it does not use $G'$ at all.  Surprisingly, in Frobenius-norm error the two GGB
methods remain better than soft-thresholding even at very high levels
of distortion.  A deeper examination of this simulation reveals that
the GGB methods are more biased but less variable than the
soft-thresholding estimates, leading to a smaller Frobenius-norm
error.  We do expect this behavior to be problem-specific and in
general would expect GGB methods with highly distorted graphs to
perform worse than soft-thresholding. 

\begin{figure}
  \centering
  \includegraphics[width=0.45\linewidth]{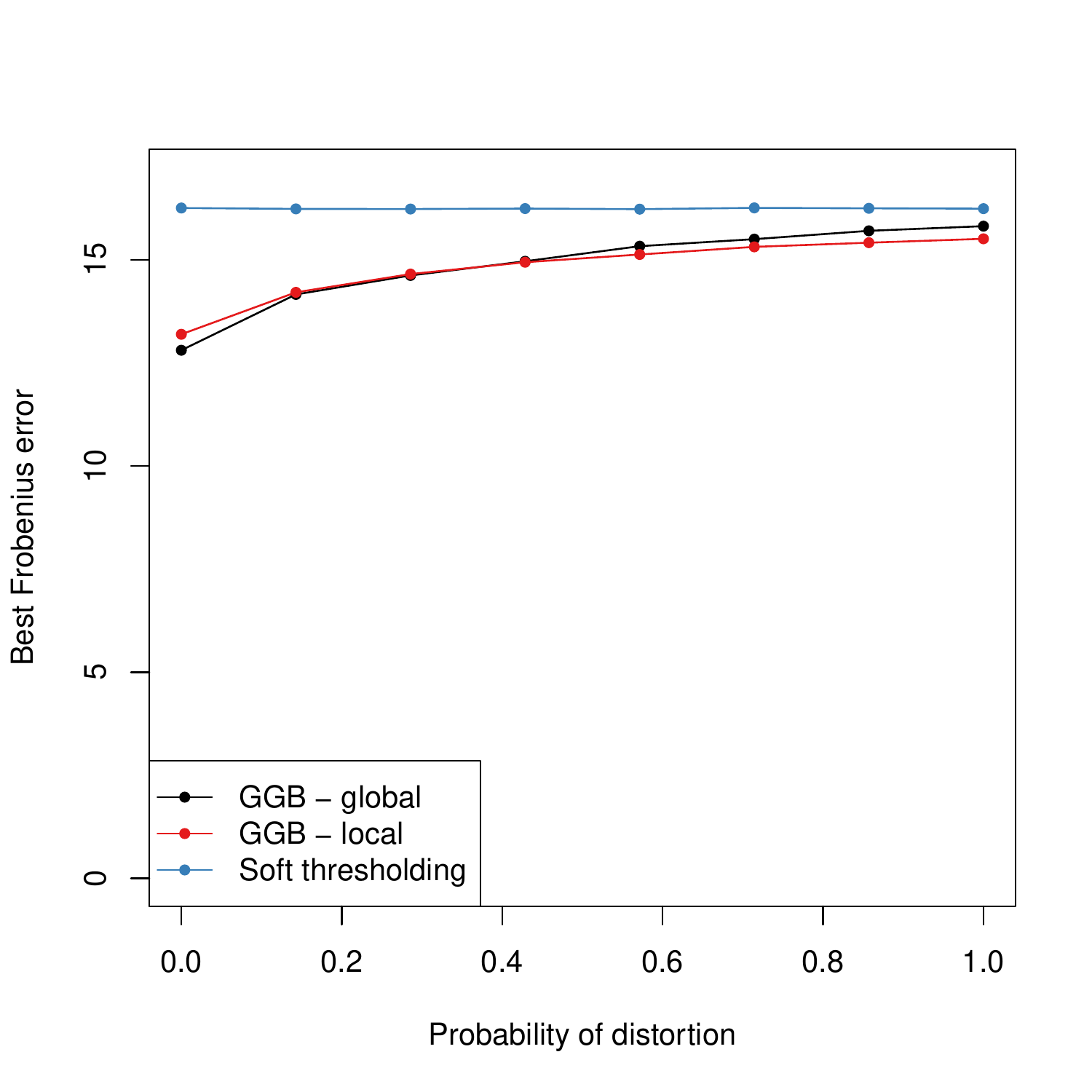}
  \includegraphics[width=0.45\linewidth]{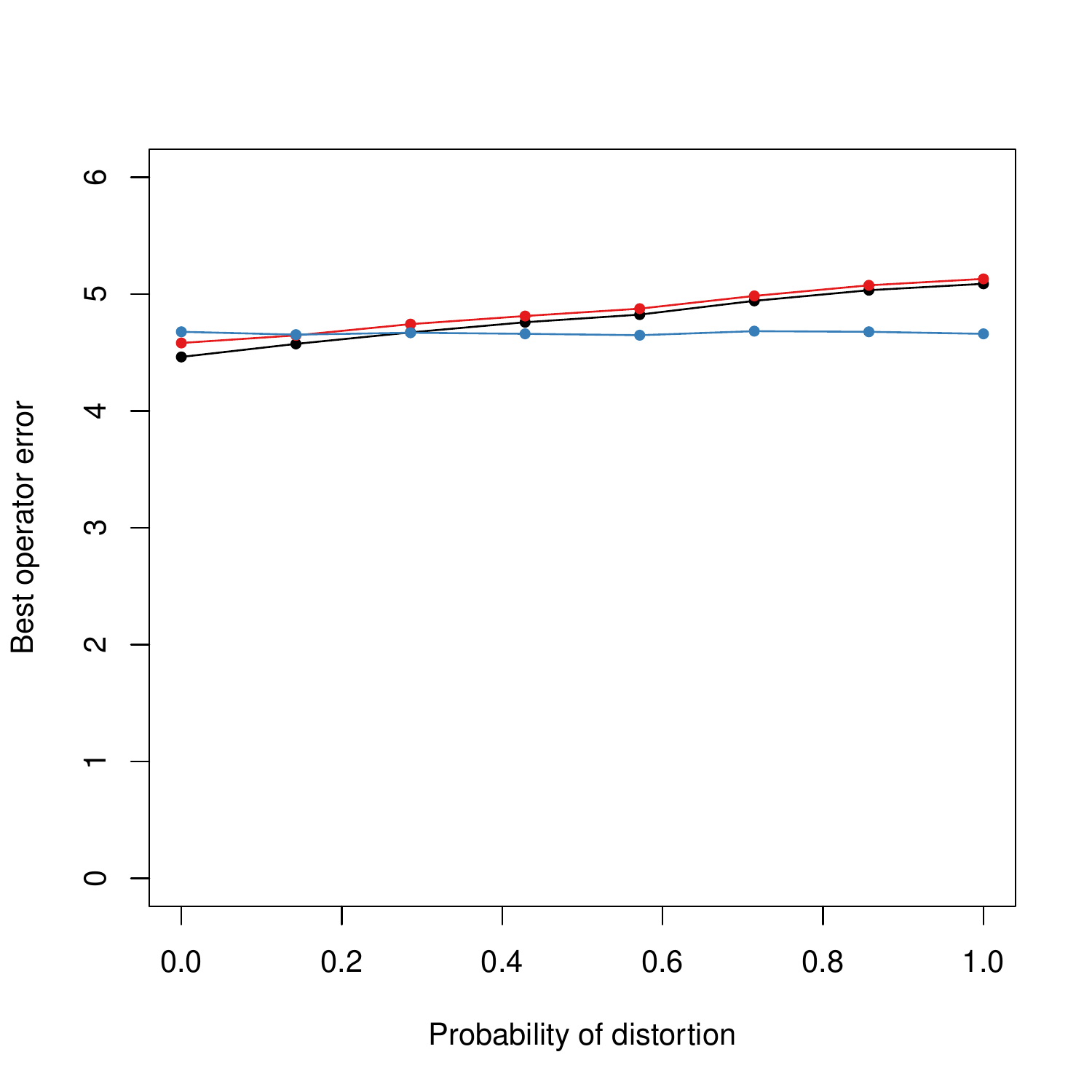}
  \caption{\small\em The effect of graph misspecification on the performance of
    the global and local graph-guided banding methods in terms of
    Frobenius and operator norms. 
}
  \label{fig:distortion}
\end{figure}

\subsection{Application to images of handwritten digits}
\label{sec:image-data}

In this section, we apply our estimators to a dataset of ($16\times16$)-pixel images of handwritten digits
\citep{le1990handwritten,ESL}\footnote{available at \url{http://statweb.stanford.edu/~tibs/ElemStatLearn/datasets/zip.digits/}}.  There are between (roughly) 700 and 1500 images per
digit type.  For each digit type, we take a random $n = 90$ images to form the ``training set'' sample covariance matrix $\S_{\text{train}}$ and use the
remaining data to form the ``test set'' sample covariance matrix
$\S_{\text{test}}$ (our results are based on 20 replications of this process).  A natural seed graph in this situation is the
two-dimensional lattice graph formed by connecting adjacent pixels.
We compare our two GGB estimators to the soft-thresholding estimator.
 We fit each covariance estimator to the
training data and compute its best attainable (over choice of tuning parameter) distance in Frobenius norm to $\S_{\text{test}}$,
$$
\min_{\lambda\ge0}\|\hSigma_\lambda(\S_{\text{train}}) - \S_{\text{test}}\|_F^2.
$$
Table \ref{tab:digits} reports this quantity, averaged over the 20 random
train-test splits of the dataset, for each digit. \ggbvb~has a lower test error than
  soft-thresholding in all 20 replicates for the digits 2-9;
  \ggbgb~has a lower test error than
  soft-thresholding in all 20 replicates for all digits other than 0, 1, and
  6. In fact, there is no digit for which soft-thresholding wins in a
  majority of replicates.
This indicates that exploiting the neighborhood structure
inherent to image data is useful for covariance estimates in this context.

\begin{table}
\scriptsize
\caption{\small\em A comparison of best attainable Frobenius-norm test error
  (averaged over 20 replicates).  Standard errors are reported in
  parentheses.   For all but two of the digits, a GGB method has lower test
  error than soft-thresholding in all 20 out of 20 replications.}
\centering
\begin{tabular}[t]{l|l|l|l}
\hline
  & \ggbgb & \ggbvb & Soft thresholding\\
\hline
Digit 0 & 10.8 (0.1) & 10.8 (0.1) & 11.0 (0.1)\\
\hline
Digit 1 & 2.6 (0.1) & 2.6 (0.1) & 2.6 (0.1)\\
\hline
Digit 2 & 10.9 (0.2) & 10.9 (0.2) & 11.3 (0.1)\\
\hline
Digit 3 & 8.9 (0.1) & 8.9 (0.1) & 9.3 (0.1)\\
\hline
Digit 4 & 9.4 (0.1) & 9.5 (0.1) & 9.9 (0.1)\\
\hline
Digit 5 & 10.7 (0.2) & 10.8 (0.2) & 11.2 (0.2)\\
\hline
Digit 6 & 9.2 (0.2) & 9.1 (0.2) & 9.2 (0.2)\\
\hline
Digit 7 & 7.3 (0.1) & 7.4 (0.1) & 7.5 (0.1)\\
\hline
Digit 8 & 9.3 (0.1) & 9.4 (0.1) & 9.8 (0.1)\\
\hline
Digit 9 & 7.8 (0.2) & 7.9 (0.2) & 8.0 (0.2)\\
\hline
\end{tabular}
\label{tab:digits}
\end{table}

\section{Discussion}
\label{sec:discussion-1}

This paper generalizes the notion of banded covariance estimation to
situations in which there is a known underlying graph structure to the
variables being measured.  The traditional notion of a banded matrix
corresponds to the special case where this graph is a path encoding
the ordering of the variables.  We propose two new ``graph-guided'' estimators, based on
convex optimization problems: $\ggbgb$ enforces a single bandwidth
across all variables while $\ggbvb$ allows each variable to have
its own bandwidth.  There are many interesting avenues of future research.  In this paper, we assume the seed graph is given, and the
challenge is in determining the bandwidth(s).  One could imagine
situations in which the seed graph itself is unavailable (but in which
we do believe that such a graph exists).  In such a setting, can one
identify the smallest seed graph $G$ for which the covariance matrix's
sparsity pattern is banded with respect to $G$?  This generalizes the
problem of discovering an ordering for which one can apply a
traditional known-ordering estimator considered by
\citet{wagaman2009discovering,zhu2009estimating}.  A simple heuristic
strategy might involve forming a minimum spanning tree based on the
complete graph with edge weights given by the sample covariance
matrix.  The graph-guided banding estimators would then have sparsity
patterns with this minimum spanning tree seed graph serving as backbone.  Another future avenue of research is
in applying the graph-guided structure to the inverse covariance matrix.
 Finally, we have seen (e.g., in Theorem~\ref{thm:ggb} and in the
 empirical study) that the
 properties of the seed graph $G$ influence the success of the GGB
 estimators.  More generally understanding how the graph properties of $G$ affect
 the difficulty of the estimation problem itself would be a particularly interesting direction of
 future study.

\appendix
\section{Proof of Lemma \ref{lem:bcd-for-psd} of main paper}
\label{supp-sec:bcd-derivation}

 \begin{lem}
   Consider the problem
   \begin{align*}
     \hSigma=\arg\min_{\bSigma\succeq\delta
       I_p}\left\{\half\|\S-\bSigma\|_F^2+\lambda\Omega(\bSigma)\right\},
   \end{align*}
where $\Omega$ is a closed, proper convex function, and let
$$
\text{Prox}_{\lambda\Omega}({\bf M}) = \arg\min_{\bSigma}\left\{\half\|{\bf M}-\bSigma\|_F^2+\lambda\Omega(\bSigma)\right\}
$$
denote the proximal operator of $\lambda\Omega(\cdot)$ evaluated at a
matrix ${\bf M}$.
Algorithm \ref{alg:bcd-dual}
converges to $\hSigma$.
\label{lem:bcd-for-psd2}
 \end{lem}
 \begin{proof}
Problem \eqref{eq:general} is equivalent to 
   \begin{align*}
     \min_{\bSigma,\tSigma}\left\{\half\|\S-\bSigma\|_F^2+\lambda\Omega(\tSigma)\st\bSigma=\tSigma,\quad\bSigma\succeq\delta
       I_p \right\},
   \end{align*}
which has Lagrangian 
$$
L(\bSigma,\tSigma;\B,\C)=\half\|\S-\bSigma\|_F^2+\lambda\Omega(\tSigma)+\langle\bSigma-\tSigma,\B\rangle+\langle\delta I_p-\bSigma,\C\rangle.
$$ 
The dual function is given by minimizing over $\bSigma$ and $\tSigma$:
\begin{align*}
  \inf_{\bSigma,\tSigma}L(\bSigma,\tSigma;\B,\C)&=\delta\langle
  I_p,\C\rangle+\inf_{\bSigma}
  \left\{\half\|\S-\bSigma\|_F^2+\langle\bSigma,\B-\C\rangle\right\}+\inf_{\tSigma}\left\{\lambda\Omega(\tSigma)-\langle\tSigma,\B\rangle\right\}\\
&=\delta\langle
  I_p,\C\rangle-\half\|\B-\C-\S\|_F^2+\half\|\S\|_F^2-[\lambda\Omega]^*(\B)
\end{align*}
where $[\lambda\Omega]^*$ is the convex conjugate of $\lambda\Omega$. Also,
$$
\hSigma=\S-\hat{\B}+\hat{\C}.
$$
The dual optimal $(\hat\B,\hat\C)$ is given by solving
$$
\min_{\B,\C}\half\|\B-\C-\S\|_F^2+[\lambda\Omega]^*(\B)-\delta\langle
  I_p,\C\rangle\st \C\succeq 0,
$$
which is a convex problem.  Furthermore, by \citet{Tseng01}, it can be
solved by alternately optimizing over $\B$ (holding $\C$ fixed) and
over $\C$ (holding $\B$ fixed).  The solution to the optimization over
$\B$ is the proximal operator of $[\lambda\Omega]^*$ evaluated at $\C+\S$.  The Moreau
decomposition \citep[see, e.g.,][]{parikh2014proximal} for the closed,
proper convex function $\lambda\Omega$ gives us 
$$
\hat\B(\C)=\text{Prox}_{[\lambda\Omega]^*}({\bf \S+\C})=\S+\C -\text{Prox}_{\lambda\Omega}({\bf \S+\C}).
$$
Optimizing over $\C$ while holding $\B$ fixed requires solving
$$
\hat\C(\B)=\min_{\C}\half\|\B-\C-\S\|_F^2-\delta\langle
  I_p,\C\rangle\st \C\succeq 0,
$$
which in terms of the eigenvalue decomposition of $\B-\S=\U\Lambda\U^T$
becomes
$$
\hat\C(\B)=\min_{\C}\half\|\C-\U\Lambda\U^T\|_F^2-\delta\langle
  I_p,\C\rangle\st \C\succeq 0
$$
or $\hat\C(\B)=\U\hat\D(\B)\U^T$ where
$$
\hat\D(\B)=\min_{\D}\half\|\D-\Lambda\|_F^2-\delta\langle
  I_p,\D\rangle\st \D\succeq 0,
$$
which is solved by 
$$
\hat\D(\B)=\diag([\Lambda_{ii}+\delta]_+).
$$
 \end{proof}

\section{Proof of Theorem \ref{thm:ggb} of main paper (Frobenius
  error)}
\label{supp-app:proof-ggb}
\begin{proof}

Since $\delta=-\infty$ and $M=\diam(G)$, the estimator is more simply
\begin{align}
\hSigma:=\hSigma_\lambda=\arg\min_{\bSigma}\left\{\half\|\S-\bSigma\|_F^2+\lambda P(\bSigma;G)\right\}
\label{eq:pr-pre}
\end{align}
where
$$
P(\bSigma;G)=\min_{\{\V[b]\in\real^{p\times
    p}\}}\left\{\sum_{b=1}^Mw_b\|\V[b]\|_F\st\bSigma^-=\sum_{b=1}^M\V[b]_{g_b}\right\}.
$$
In what follows, whenever we write a sum with subscript $b$, it is
understood to run along $b=1,\ldots,M$; this also holds for the $\sV$ notation, which is
short for $\sV_{b=1}^M$.

We begin by observing that \eqref{eq:pr-pre} is a strongly convex problem in
$\bSigma$, and thus there is a unique minimizer $\hSigma$, although the
minimizing set of $\hV[b]$'s may not be unique.  Also, since
$\diag(\bSigma)$ does not appear in the penalty, we see that
$\diag(\hSigma)=\diag(\S)$.  Thus, we can write \eqref{eq:pr-pre} as
\begin{align}
\hSigma^-=\arg\min_{\bSigma}\left\{\half\|\S^--\bSigma\|_F^2+\lambda P(\bSigma;G)\right\}.
\label{eq:pr0}
\end{align}
We can write \eqref{eq:pr0} in terms of the $\V[b]$'s:
\begin{align}
\shV\in\arg\min_{\V[b]}\left\{\half\|\S-\sum_{b}\V\|_F^2+\lambda\sum_{b}w_{b}\|\V\|_F\st\V_{g_{b}^c}=0~\forall
b\right\}.
  \label{eq:pr}
\end{align}
Note that we have replaced $\S^-$ by $\S$ (for notational convenience)
since $g_b$ does not include diagonal elements.
The lemma below provides necessary and sufficient conditions for
$\shV$ to be a solution to \eqref{eq:pr}.
\begin{lem}
Writing $\hSigma=\sum_{b}\hV$, $\shV$ is a solution to \eqref{eq:pr}
iff. for each $b\in\{1,\ldots,M\}$, $\hV_{g_{b}^c}=0$ and 
$$
[\S-\hSigma]_{g_{b}}=\frac{\lambda w_{b} \hV_{g_{b}}}{\|\hV\|_F}\quad\text{ if }\hV\neq0
$$
and 
$
\|[\S-\hSigma]_{g_{b}}\|_F\le\lambda w_{b}
$
otherwise.
\label{lem:kkt}
\end{lem}
\begin{proof}
Problem \eqref{eq:pr} is equivalent to a non-overlapping group lasso
problem, and this lemma provides the KKT conditions of this problem,
which are well-known.
\end{proof}

\begin{cor}
Let $\shV$ be a solution to \eqref{eq:pr}.  Then,
$$
\|[\S-\hSigma]_{g_{b}}\|_F< \lambda w_{b}\implies \hV=0.
$$

\label{cor:strict}
\end{cor}

Suppose 
$\supp(\bSigma^*)=g_{B^*}$.
If we can construct a solution $\shV$ to \eqref{eq:pr} that has
$$
\max_{b>B^*} w_{b}^{-1}\|[\S-\hSigma]_{g_{b}}\|_F< \lambda,
$$
then, by Corollary \ref{cor:strict}, we know that $\hV=0$ for all
$b>B^*$ or 
$$
\{b:\hV\neq0\} \subseteq\{1,\ldots,B^*\}.
$$
In fact, the corollary implies that this holds for {\em all}
solutions, not just for our constructed $\shV$.

Consider the restricted problem
$$
\sbV_{b=1}^{B^*}\in\arg\min_{\V}\left\{\half\|\S-\sum_{b=1}^{B^*}\V\|_F^2+\lambda\sum_{b=1}^{B^*}
  w_{b}\|\V\|_F\st\V_{g_{b}^c}=0\text{ for }1\le b\le B^*\right\},
$$
and write $\bbSigma=\sum_{b=1}^{B^*}\bV$.
This restricted problem is of an identical form to \eqref{eq:pr} but
with $M$ replaced by $B^*$.  Thus, Lemma \ref{lem:kkt} applies if we
replace $b=1,\ldots,M$
with $b=1,\ldots,B^*$:  For each $b\le B^*$,
$$
[\S-\bbSigma]_{g_{b}}=\frac{\lambda w_{b} \bV_{g_{b}}}{\|\bV\|_F}\quad\text{ if }\bV\neq0
$$
and 
$
\|[\S-\bbSigma]_{g_{b}}\|_F\le\lambda w_{b}
$
otherwise.

Now, for $b>B^*$, let $\bV=0$.  We show that $\sbV_{b=1}^M$
is a solution to \eqref{eq:pr}.  Clearly, $\sum_{b=1}^M\bV=\sum_{b=1}^{B^*}\bV=\bbSigma$, so by Lemma \ref{lem:kkt} it
simply remains to show that
$$
\|[\S-\bbSigma]_{g_{b}}\|_F\le\lambda w_{b}\text{ for }b>B^*.
$$
For $b>B^*$: 
\begin{align*}
  \|(\S-\bbSigma)_{g_{b}}\|_F^2&=\|(\S-\bbSigma)_{g_{B^*}}\|_F^2+\|\S_{s_{(B^*+1):b}}\|_F^2\\
&\le\lambda^2w_{B^*}^2+\|(\S-\bSigma^*)_{s_{(B^*+1):b}}\|_F^2\\
&\le\lambda^2w_{B^*}^2+|s_{(B^*+1):b}|\cdot\|\S-\bSigma^*\|_\infty^2.
\end{align*}
Now, on the event 
\begin{align}
  \calA_\lambda=\{\|\S-\bSigma^*\|_\infty<\lambda\},\label{eq:Alam}
\end{align}
we have
$$
  \|(\S-\bbSigma)_{g_{b}}\|_F^2<\lambda^2 (w_{B^*}^2+|s_{(B^*+1):b}|)=\lambda^2 w_b^2,
$$
using that $w_b^2=|g_b|=|g_{B^*}|+|s_{(B^*+1):b}|$.
This establishes, by Lemma \ref{lem:kkt}, that $\sbV$ is a solution to
\eqref{eq:pr} and, by Corollary \ref{cor:strict} that
$\{b:\hV\neq0\}\subseteq\{1,\ldots,B^*\}$ for all solutions $\shV$ to
\eqref{eq:pr}.  Given the support constraints of each $\hV$, it
follows that the overall $\hSigma$ has bandwidth no larger than $B^*$.

\begin{prop}
If $w_b=|g_b|^{1/2}$, $\delta=-\infty$, and $M=\diam(G)$, then
$B(\hSigma_\lambda)\le B^*$
on $\calA_\lambda$.
\end{prop}

\noindent Having established the above proposition about the bandwidth, it is
straightforward to prove a result about the estimation error. On $\calA_\lambda$,
$$
\|(\hSigma-\bSigma^*)_{g_{B^*}}\|_F\le\|(\S-\bSigma^*)_{g_{B^*}}\|_F+\|(\S-\hSigma)_{g_{B^*}}\|_F\le\|\S-\bSigma^*\|_\infty\cdot
|g_{B^*}|^{1/2}+\lambda w_{B^*}<2\lambda w_{B^*}
$$
and
$\hSigma_{g_{B^*}}=0$.
 Thus,
 \begin{align}
\|\hSigma-\bSigma^*\|_F^2=\|\diag(\S)-\diag(\bSigma^*)\|_F^2+\|(\hSigma-\bSigma^*)_{g_{B^*}}\|_F^2<\lambda^2(p+4w_{B^*}^2).\label{eq:withdiag}
 \end{align}
\begin{prop}
 If $w_b=|g_b|^{1/2}$, $\delta=-\infty$, and $M=\diam(G)$, then
$$
\|\hSigma-\bSigma^*\|_F^2\le\lambda^2(p+4|g_{B(\bSigma^*)}|)
$$
on $\calA_\lambda$.
\end{prop}  

Under Assumptions \ref{ass:distribution} and
\ref{ass:scaling} of the main paper,
there exists a constant $c>0$ for which 
\begin{align}
  P(\calA_\lambda)\ge1-c/\max\{p,n\}
\label{eq:Alam-prob}
\end{align}
 if we take
$\lambda=x\sqrt{\log\max\{p,n\}/n}$ with $x$ sufficiently large.  This
follows directly from bound (A.4.1) of Theorem A.1 of \citet{Bien15supp}.  Theorem
\ref{thm:ggb} of the main paper follows.
\end{proof}

\section{Proof of Theorem \ref{thm:bandwidth-recovery} of main
  paper (bandwidth recovery)}
\label{supp-app:bandwidth-recovery}

We prove a result that is a bit more general than Theorem
\ref{thm:bandwidth-recovery} of the main paper. 
\begin{theorem}
Under the minimum signal condition 
$$
\min_{b\le B^*}\|\bSigma^*_{g_b}\|_F/|g_b|^{1/2} > 2\lambda,
$$
and the conditions of Theorem \ref{thm:ggb} of the main paper,
$B(\hSigma)=B(\bSigma^*)$ with probability at least $1-c/\max\{p,n\}$.
\label{thm:supp-bandwidth}
\end{theorem}
\begin{proof}
On the event $\calA_\lambda$ defined in \eqref{eq:Alam}, for $b\le B^*$,
\begin{align*}
 \|(\hSigma-\bSigma^*)_{g_b}\|_F&\le\|(\S-\bSigma^*)_{g_b}\|_F+\|(\S-\hSigma)_{g_b}\|_F\\
&<\lambda |g_b|^{1/2}+\|(\S-\hSigma)_{g_b}\|_F\\
&\le\lambda |g_b|^{1/2}+\lambda w_b=2\lambda w_b.
\end{align*}
Thus,
$$
\max_{b\le B^*}w_b^{-1}\|(\hSigma-\bSigma^*)_{g_b}\|_F<2\lambda.
$$
 
We already know that $B(\hSigma)\le B^*$, so it suffices to
show that
$$
\min_{b\le B^*}\|\hSigma_{g_b}\|_F/w_b>0,
$$  
which follows since
\begin{align*}
  \|\hSigma_{g_b}\|_F/w_b\ge
  \|\bSigma^*_{g_b}\|_F/w_b-\|(\hSigma-\bSigma^*)_{g_b}\|_F/w_b>  2\lambda-2\lambda=0.
\end{align*}
The result then follows from \eqref{eq:Alam-prob}.
\end{proof}
\begin{cor}
  Under the minimum signal condition 
$$
\min_{jk\in g_{B^*}}|\bSigma^*_{jk}| > 2\lambda,
$$
and the conditions of Theorem \ref{thm:ggb} of the main paper, $B(\hSigma)=B(\bSigma^*)$
with probability at least $1-c/\max\{p,n\}$.
\label{cor:elem-beta-min}
\end{cor}
\begin{proof}
  For any $b\le B^*$,
$$
\|\bSigma^*_{g_b}\|_F^2/|g_b|=
\sum_{jk\in g_{b}}[\bSigma^*_{jk}]^2/|g_b|>\sum_{jk\in g_{b}}4\lambda^2/|g_b|=4\lambda^2.
$$
Thus, the conditions of Theorem \ref{thm:supp-bandwidth} hold.
\end{proof}

\section{Proof of Theorem \ref{thm:slowrate} (Frobenius error of local bandwidths estimator)}
\label{supp-sec:slowrate}

\begin{proof}
  
Let $\tSigma:=\tSigma_\lambda$ be the minimizer of
Equation \eqref{eq:ggb-vb} of the main paper.
By definition,
$$
\half\|\S-\tSigma\|_F^2+\lambda \tilde P(\tSigma;G)\le \half\|\S-\bSigma^*\|_F^2+\lambda \tilde P(\bSigma^*;G).
$$
Since
$$
\|\S-\tSigma\|_F^2=\|\S-\bSigma^*\|_F^2+\|\tSigma-\bSigma^*\|_F^2-2\langle\S-\bSigma^*,\tSigma-\bSigma^*\rangle,
$$
we have
\begin{align}
  \half\|\tSigma-\bSigma^*\|_F^2+\lambda\tilde P(\tSigma;G)\le
  \langle\S-\bSigma^*,\tSigma-\bSigma^*\rangle+\lambda\tilde
  P(\bSigma^*;G).\label{eq:innerprod}
\end{align}
By Lemma 3 of \citet{Obozinski11}, the dual norm is $\tilde P^*({\bf B};G)=\max_{jb}w_{jb}^{-1}\|{\bf
B}_{g_{jb}}\|_F
$, and so
\begin{align}
\langle\S-\bSigma^*,\tSigma-\bSigma^*\rangle\le\lambda^2p+\tilde
  P^*(\S-\bSigma^*;G)\tilde P(\tSigma-\bSigma^*;G).
  \label{eq:innerwithdiag}
\end{align}
Define the index sets $\I=\{jb: 1\le b\le M_j\}$ and $\I^*=\{jb: 1\le b\le B_j^*\}$.

Observe that with $w_{jb}=|g_{jb}|^{1/2}$
$$
\max_{jb\in\I}w_{jb}^{-1}\|{\bf
B}_{g_{jb}}\|_F\le\|{\bf B}\|_\infty.
$$
It follows that
$$
\langle\S-\bSigma^*,\tSigma-\bSigma^*\rangle<\lambda^2p+\lambda\tilde P(\tSigma-\bSigma^*;G),
$$
as long as $\|\S-\bSigma^*\|_\infty< \lambda$, that is on the event
$\calA_\lambda$, as defined in \eqref{eq:Alam}.
Thus, combining this with \eqref{eq:innerprod}, 
$$
\half\|\tSigma-\bSigma^*\|_F^2+\lambda\tilde P(\tSigma;G)\le \left[\lambda^2p+\lambda\tilde P(\tSigma-\bSigma^*;G)\right]+\lambda\tilde P(\bSigma^*;G).
$$
holds on  $\calA_\lambda$.  Using the triangle inequality ($\tilde P$
is a norm),
$$
\half\|\tSigma-\bSigma^*\|_F^2\le\lambda^2p+ \lambda\left[\tilde P(\tSigma-\bSigma^*;G)-\tilde P(\tSigma;G)+\tilde P(\bSigma^*;G)\right]\le\lambda^2p+2\lambda\tilde P(\bSigma^*;G).
$$
By definition of $\tilde P$,
\begin{align}
  \tilde P(\bSigma^*;G)\le \sum_{jb\in\I}w_{jb}\|\V[jb]\|_F\label{eq:tildeP}
\end{align}
for any choice of $\sV[jb]$ such that $\V[jb]_{g_{jb}^c}=0$ and
$\bSigma^*=\sum_{jb}\V[jb]$.  

Consider taking
$\V[j'b']_{jk}=0$, for all $jk\in\{1,\ldots,p\}^2$ and $j'b'\in\I$, unless
\begin{enumerate}[(i)]
\item $d_G(j,k)\le\min\{B^*_j,B^*_k\}$, in which case we take
$
\V[jB_j^*]_{jk}=\V[jB_j^*]_{kj}=\V[kB_k^*]_{jk}=\V[kB_k^*]_{kj}=\half\bSigma_{jk}^*.
$
\item $B_j^*<d_G(j,k)\le B_k^*$, in which case we take
$
\V[kB_k^*]_{jk}=\V[kB_k^*]_{kj}=\bSigma_{jk}^*.
$
\item $B_k^*<d_G(j,k)\le B_j^*$, in which case we take
$
\V[jB_j^*]_{jk}=\V[jB_j^*]_{kj}=\bSigma_{jk}^*.
$
\end{enumerate}

We verify that this choice of $\sV[jb]$ satisfies the necessary
constraints.  Clearly, $\V[jb]=0$ for $b\neq B_j^*$.  If $jk\not\in
g_{jB_j^*}$, then $d_G(j,k)>B_j^*$, which excludes cases (i) and (ii),
the only two cases in which $\V[jB_j^*]_{jk}$ (and $\V[jB_j^*]_{kj}$)
could be nonzero.  Thus, $\V[jB_j^*]_{g_{jB_j^*}^c}=0$.  

We next verify that $\sum_{j'b'\in\I}\V[j'b']=\bSigma^*$.  In case
(i),
$$
\sum_{j'b'\in\I}[\V[j'b']]_{jk}=\V[jB_j^*]_{jk}+\V[kB_k^*]_{jk}=\bSigma_{jk}^*.
$$
In case (ii), 
$$
\sum_{j'b'\in\I}[\V[j'b']]_{jk}=\V[kB_k^*]_{jk}=\bSigma_{jk}^*.
$$
In case (iii), 
$$
\sum_{j'b'\in\I}[\V[j'b']]_{jk}=\V[jB_j^*]_{jk}=\bSigma_{jk}^*.
$$
Now, we have that
$|\V[jB_j^*]_{jk}|\le |\bSigma^*_{jk}|$, which implies that
$\|\V[jB_j^*]\|_F\le \|\bSigma^*_{g_{jB_j^*}}\|_F$. Thus,

$$
\sum_{jb\in\I}w_{jb}\|\V[jb]\|_F=\sum_{j}w_{jB_j^*}\|\V[jB_j^*]\|_F\le
\sum_{j}w_{jB_j^*}\|\bSigma^*_{g_{jB_j^*}}\|_F\le
\|\bSigma^*\|_\infty\sum_{j}|g_{jB_j^*}|.
$$
Now, each $jk\in S^*$ appears only in $g_{jB^*_j}$ or $g_{kB^*_k}$ (or
potentially both).  Thus, by \eqref{eq:tildeP}, $\sum_{j}|g_{jB^*_j}|\le 2|S^*|$.  It follows
that $\tilde P(\bSigma^*;G)\le2\|\bSigma^*\|_\infty|S^*|$, and so (on $\calA_\lambda$)
$$
\|\tSigma-\bSigma^*\|_F^2\le \lambda^2p +4\|\bSigma^*\|_\infty|S^*|\lambda.
$$
 Assumptions \ref{ass:distribution} and
\ref{ass:scaling} of the main paper are used (as in Appendix
\ref{supp-app:proof-ggb}) to show that $\calA_\lambda$ holds with the desired probability.
\end{proof}

\section{Proofs about GGB on sample correlation matrix}
\label{sec:proof-theorem-refggb}

We begin by establishing a concentration inequality for the
correlation matrix that will be used in the proofs of both of these theorems.  This result is similar to Lemma A.2 of
\citet{liu2017tiger} except that we assume sub-Gaussian rather than
Gaussian data.  It is also similar to Lemma D.5 of
\citet{sun2017graphical} except that we assume that $\S$ has been
centered by the sample mean.

\begin{lem}
 Suppose that Assumptions \ref{ass:distribution},
 \ref{ass:scaling}, and \ref{ass:bounded} hold.  There exists a constant $x$ (i.e., a value not depending on $n$ or $p$) such that for all $n$ and $p$ for which $\log\{p,n\}/n<1/x^2$,
\begin{align*}
\P(\|\hcorr - \corr\|_\infty> x\sqrt{\log(\max\{p,n\})/n})
&\le 16/ \max\{p,n\}.
\end{align*}
\label{lem:cor}
\end{lem}
\begin{proof}
The following proof is based on the proof of Lemma A.2 of
\citet{liu2017tiger}.  Certain details have been fleshed out and some
modifications have been made to accommodate the differing assumptions.

Let $\cA_j=\{|\S_{jj}-\bSigma^*_{jj}|\le
\bSigma^*_{jj}/2\}$.
Observing that $\cA_j=\{1/2\le\S_{jj}/\bSigma^*_{jj}\le 3/2\}$, it
follows that
\begin{align}
\cA_j\cap \cA_k&\implies (1/2)^2\le
\frac{\S_{jj}\S_{kk}}{\bSigma^*_{jj}\bSigma^*_{kk}}\le(3/2)^2\nonumber\\
&\implies 1/2\le
\left(\frac{\S_{jj}\S_{kk}}{\bSigma^*_{jj}\bSigma^*_{kk}}\right)^{1/2}\le
3/2\nonumber\\
&\implies (1/2)\pro_{jk}\le\hpro_{jk}\le
(3/2) \pro_{jk}\nonumber\\
&\implies 
|\hpro_{jk}-\pro_{jk}|\le
  (1/2)\pro_{jk}.
\label{eq:pro}
\end{align}
Now,
\begin{align*}
 \{|\hcorr_{jk} - \corr_{jk}|>t\} &=  \{|\S_{jk} - \bSigma^*_{jk}|>t\hpro\}\\
&\subseteq  \{|\S_{jk} - \bSigma^*_{jk}|>t\hpro\}\cap(\cA_j\cap\cA_k)\bigcup(\cA_j\cap\cA_k)^c.
\end{align*}
By \eqref{eq:pro} and the triangle inequality,
$$
\{|\S_{jk} - \bSigma^*_{jk}|>t\hpro\}\cap(\cA_i\cap\cA_j)\subseteq\{|\S_{jk} - \bSigma^*_{jk}|>(t/2)\pro\},
$$
so
$$
 \{|\hcorr_{jk} - \corr_{jk}|>t\}\subseteq  \{|\S_{jk} - \bSigma^*_{jk}|>(t/2)\pro\}\bigcup(\cA_j\cap\cA_k)^c.
$$
Now,
\begin{align*}
  \{\max_{j\neq k}|\hcorr_{jk} - \corr_{jk}|>t\}&\subseteq \bigcup_{j\neq
    k}\left[\{|\S_{jk} - \bSigma^*_{jk}|>(t/2)\pro\}\cup\cA_j^c\cup\cA_k^c\right]\\
&= \bigcup_{j\neq
    k}\{|\S_{jk} - \bSigma^*_{jk}|>(t/2)\pro\}\cup \bigcup_{j}\cA_j^c\\
&\subseteq \bigcup_{j\neq
    k}\{|\S_{jk} - \bSigma^*_{jk}|>(t/2)\kappa\}\cup \bigcup_{j}\cA_j^c\\
&\subseteq\{\max_{j\neq k}|\S_{jk} - \bSigma^*_{jk}|>(t/2)\kappa\}\cup \{\max_j |\S_{jj} - \bSigma^*_{jj}|>(1/2)\kappa\}\\
&\subseteq\{\|\S - \bSigma^*\|_\infty>(\kappa/2)\min\{t,1\}\}.
\end{align*}
Since $\hcorr_{jj}=1=\corr_{jj}$, $\max_{j\neq k}|\hcorr_{jk} -
\corr_{jk}|=\|\hcorr-\corr\|_\infty$ and thus we have established that
$$
\P(\|\hcorr-\corr\|_\infty>t)\le \P(\|\S - \bSigma^*\|_\infty>(\kappa/2)\min\{t,1\}).
$$

By Lemma A.1 of \citet{Bien15supp} (with a typo corrected), there exist constants $c_1,c_2$
such that
$$
\P(\|\S-\bSigma^*\|_\infty > \epsilon)\le 2p^2e^{-c_2n(\epsilon/\bar\kappa)^2}+8pe^{-c_1n \epsilon/\bar\kappa}
$$
for any $0<\epsilon<2\bar\kappa$.  This constraint on $\epsilon$ implies
that $(\epsilon/\bar\kappa)^2< 2 (\epsilon/\bar\kappa)$ and so
$$
2p^2e^{-c_2n(\epsilon/\bar\kappa)^2}+8pe^{-c_1n \epsilon/\bar\kappa}<2p^2e^{-c_2n(\epsilon/\bar\kappa)^2}+8pe^{-c_1n(\epsilon/\bar\kappa)^2/2}.
$$
Taking $c_3=\min\{c_2,c_1/2\}$,
\begin{align*}
\P(\|\S-\bSigma^*\|_\infty > \epsilon)&\le (2p^2+8p)e^{-c_3n(\epsilon/\bar\kappa)^2}\le 16p^2e^{-c_3n(\epsilon/\bar\kappa)^2}.
\end{align*}
Observing that $\epsilon=(\kappa/2)\min\{t,1\}<\bar\kappa<2\bar\kappa$,
we have that
$$
\P(\|\hcorr-\corr\|_\infty>t)\le 16p^2e^{-c_3n\kappa^2/(2\bar\kappa)^2\min\{t^2,1\}}.
$$
We take $t=x\sqrt{\log(\max\{p,n\})/n}$ for $x=12(\bar\kappa/\kappa)^2/c_3$. By Assumption \ref{ass:bounded},
$\kappa$ and $\bar\kappa$ are treated as constants.
If $x^2\log(\max\{p,n\})/n$ then $\min\{t^2,1\}=t^2$ and 
$$
\P(\|\hcorr-\corr\|_\infty>t)\le 16p^2e^{-c_3\kappa^2/(2\bar\kappa)^2x^2 \log(\max\{p,n\})}\le 16\max\{p,n\}^{2-c_3\kappa^2/(2\bar\kappa)^2x^2}.
$$
Thus, we just need that $2-c_3\kappa^2/(2\bar\kappa)^2x^2\le -1 $,
which can be ensured by choosing
$$
x= 12(\bar\kappa/\kappa)^2/c_3.
$$
\end{proof}

We define the event 
\begin{align}
\mathcal B_\lambda=\{\|\hcorr-\corr\|_\infty<\lambda\}\label{eq:Blam}.
\end{align}
Lemma \ref{lem:cor} establishes that there exist constants $x$ and $c$
for which choosing $\lambda= x\sqrt{\frac{\log(\max\{p,n\})}{n}}$
ensures that
$\P(\mathcal B_\lambda)\ge1-c/\max\{n,p\}$.

\subsection{Proof of Theorem \ref{thm:ggb-corr} (Frobenius error of
  global GGB on correlation matrix)}

Let $\hSigma_\lambda(\hcorr)$ be the minimizer of
Equation \eqref{eq:ggb} of the main paper where $\S$ is
replaced by $\hcorr$.

\begin{proof}[Proof of Theorem \ref{thm:ggb-corr}]
  The proof is nearly identical to that of Theorem \ref{thm:ggb}
  given in Section \ref{supp-app:bandwidth-recovery}.  We take $\hSigma$ to
  be $\hSigma_\lambda(\hcorr)$, we replace $\S$ with $\hcorr$, and we
  replace $\calA_\lambda$ with $\mathcal B_\lambda$.  Instead of
  \eqref{eq:withdiag}, we get
$$\|\hcorr-\corr\|_F^2=\|(\hSigma-\bSigma^*)_{g_{B^*}}\|_F^2<4\lambda^2w_{B^*}^2,$$
    which leads to the stated result.
    \end{proof}

\subsection{Proof of Theorem \ref{thm:slowrate-corr} (Frobenius error of
  local GGB on correlation matrix)}

Let $\tSigma_\lambda(\hcorr)$ be the minimizer of
Equation \eqref{eq:ggb-vb} of the main paper where $\S$ is
replaced by $\hcorr$.

\begin{proof}[Proof of Theorem \ref{thm:slowrate-corr}]

The proof is nearly identical to
that of Theorem \ref{thm:slowrate} given in Section
\ref{supp-sec:slowrate}.  We take $\tSigma$ to be
$\tSigma_\lambda(\hcorr)$ and replace $\S$ with $\hcorr$.
Equation \ref{eq:innerwithdiag} becomes instead
$$\langle\hcorr-\corr,\tSigma-\corr\rangle\le\tilde
  P^*(\hcorr-\corr;G)\tilde P(\tSigma-\corr;G).
$$ 
Following the same steps as in that proof, we note that on the set
$\mathcal B_\lambda$, defined in \eqref{eq:Blam},
$$
\|\tSigma-\bSigma^*\|_F^2\le 4\|\bSigma^*\|_\infty|S^*|\lambda.
$$
\end{proof}

\section{Traveling salesman problem}
\label{supp-sec:tsp}

Let $C$ be a $p\times p$ matrix of costs, with $C_{jk}$ being the cost
of traveling from node $j$ to node $k$.  The traveling salesman
problem (TSP) seeks a loop $v_1\to v_2\to\dots v_p\to v_1$ through all the
nodes such that $\sum_{\ell=1}^{p-1}C_{v_\ell v_{\ell+1}}+C_{v_{p}v_1}$ is minimal.

In our context, we have a graph $G$ and wish to find an ordering of
the nodes that is strongly informed by $G$.  In particular, we would
like this path to travel along edges of $G$ as much as possible and
when this is not possible and the path must ``fly'' between two nodes
that are not connected by a node, that it take the shortest flight
possible.  We therefore take
$$
C_{jk}=d_G(j,k)-1.
$$
We use the R package {\tt TSP} \citep{tsp}, which provides an
interface to the {\tt Concorde} solver \citep{concorde}, and Figure
\ref{fig:tsp} shows the resulting paths.
\begin{figure}
  \centering
  \includegraphics[width=0.45\linewidth]{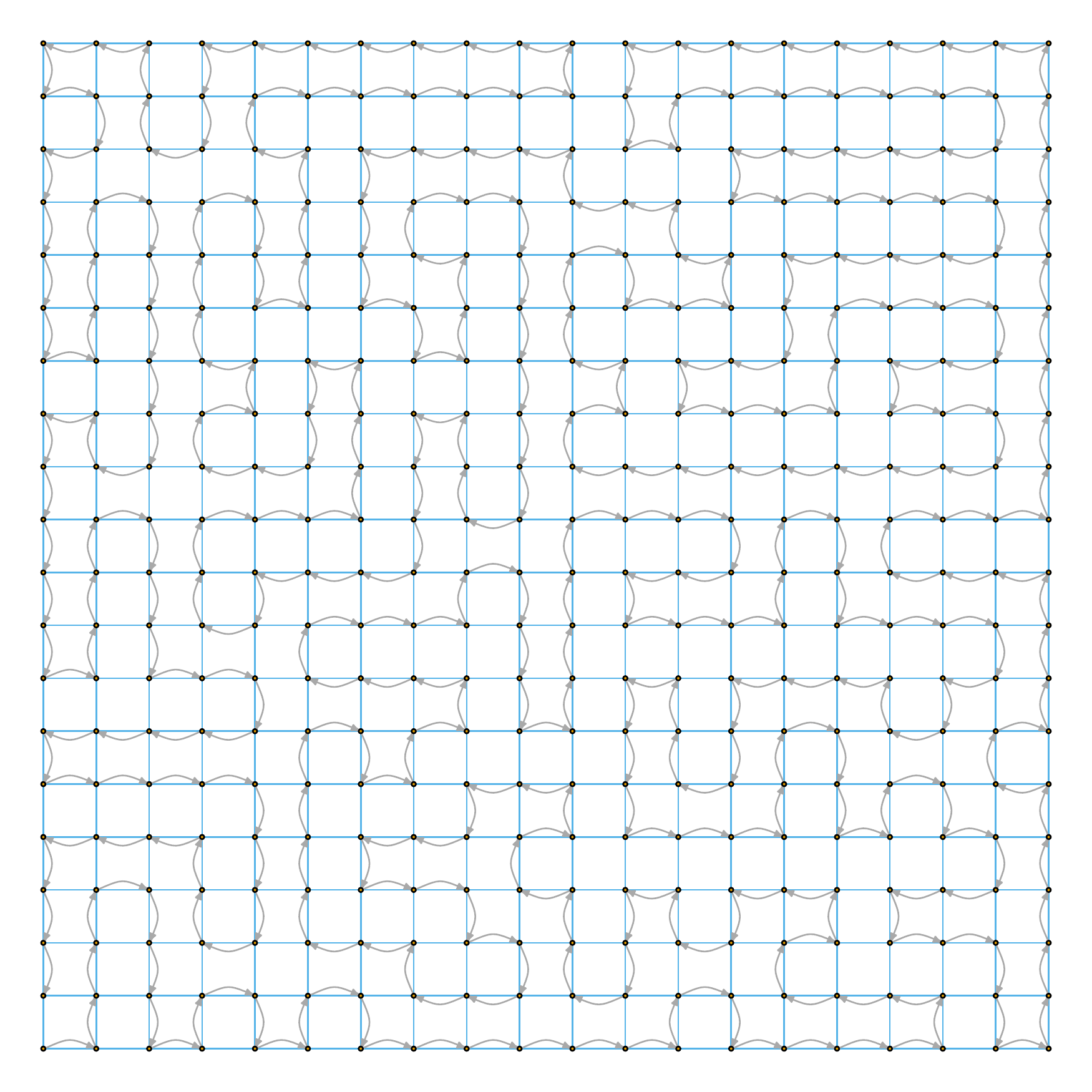}
  \includegraphics[width=0.45\linewidth]{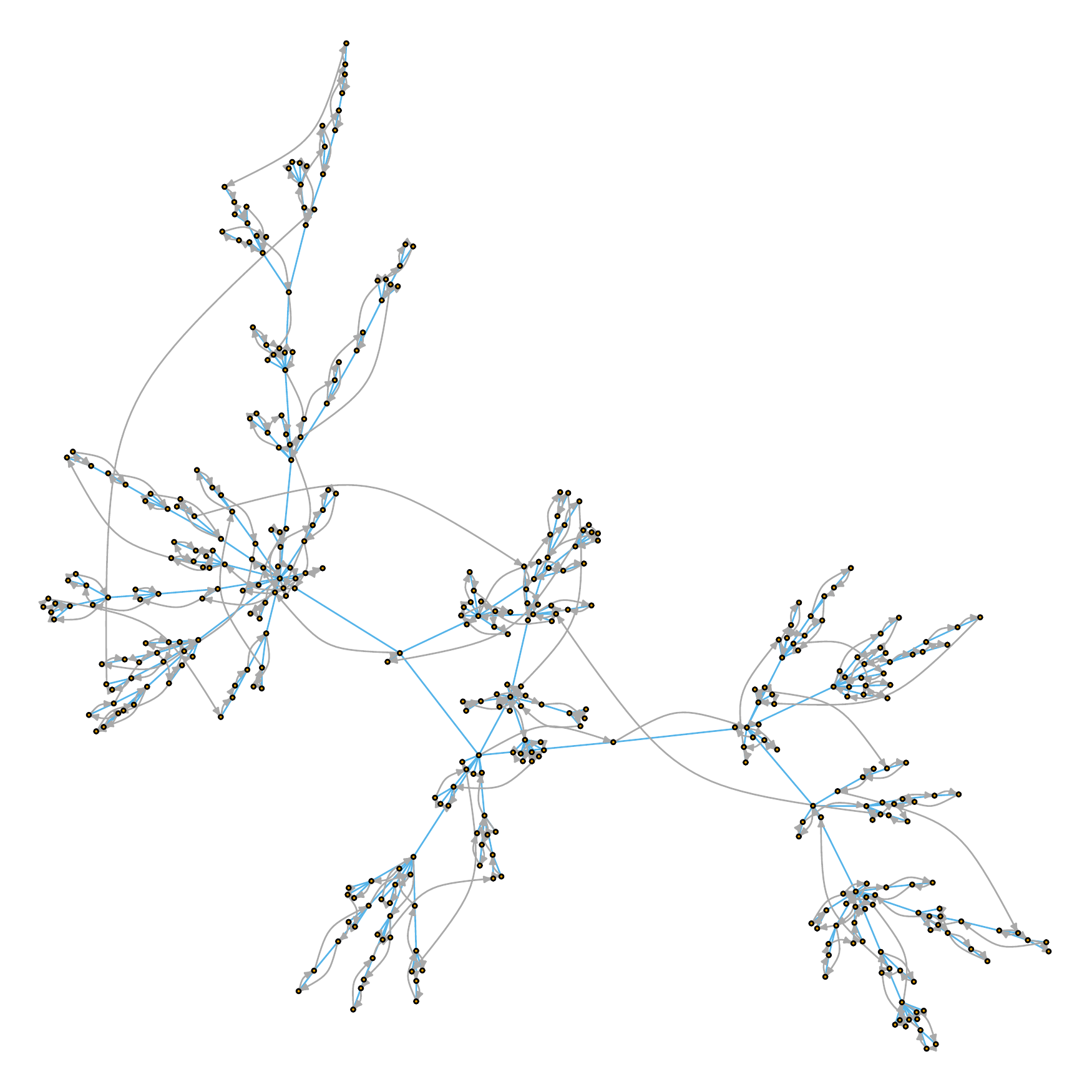}
  \caption{TSP solutions used for the two graphs used in the simulation studies}
  \label{fig:tsp}
\end{figure}

\section{An empirical study of positive semidefiniteness}
\label{supp-sec:psd}

For each of the models considered in Section
\ref{sec:lattice-versus-scale} of the main paper, we record the proportion of the time
each method is positive semidefinite.  In particular, this average is
over both the 200 simulation replicates and the values of the tuning parameter of
each method.  Table \ref{supp-tab:psd} shows that throughout this study the
banding methods (GGB and hierband) are always observed to give
positive semidefinite estimates.  This is not the case for the
soft-thresholding estimator.

\begin{table}

\caption{\label{supp-tab:psd}A comparison of propor PSD (averaged over
  a grid of tuning parameters and 200 replicates).}
\centering
\scriptsize
\begin{tabular}[t]{l|l|l|l|l}
\hline
  & GGB - global & GGB - local & Soft thresholding & TSP + hierband\\
\hline
2-d lattice (20 by 20), b = 4 & 1.000 (0.000) & 1.000 (0.000) & 0.349 (0.008) & 1.000 (0.000)\\
\hline
2-d lattice (20 by 20) w/ var-bw & 1.000 (0.000) & 1.000 (0.000) & 0.625 (0.008) & 1.000 (0.000)\\
\hline
Scale free (p = 400), b = 4 & 1.000 (0.000) & 1.000 (0.000) & 0.998 (0.001) & 1.000 (0.000)\\
\hline
Scale free (p = 400) w/ var-bw & 1.000 (0.000) & 1.000 (0.000) & 1.000 (0.000) & 1.000 (0.000)\\
\hline
\end{tabular}\end{table}

\section{A higher dimensional example}
\label{sec:higher-dim}

In Section \ref{sec:lattice-versus-scale}, we study four models in
which $p=400$ and $n=300$.  The covariance matrix has $p(p+1)/2\approx80,000$
parameters.  Thus, the number of parameters to be estimated greatly
exceeds the number of samples.  Nonetheless, we can ask what happens
if $p$ is even larger than this.  We revisit the example of the
2-d lattice, increasing the side length from 20 to 50 so that
$p=2,500$.  In this case the covariance matrix has
$p(p+1)/2\approx3,000,000$ parameters.  We focus in particular on the
GGB methods and soft thresholding.  Figure \ref{fig:largerp-roc}
provides a side-by-side comparison of the ROC curves (five replicates
are shown per method) when one
increases from $p=400$ to $p=2,500$.  The curves in the two plots are
different although this is difficult to discern by eye.  This suggests
that we are effectively in the same regime in these two simulated
scenarios.

\begin{figure}
  \centering
  \includegraphics[width=0.8\linewidth]{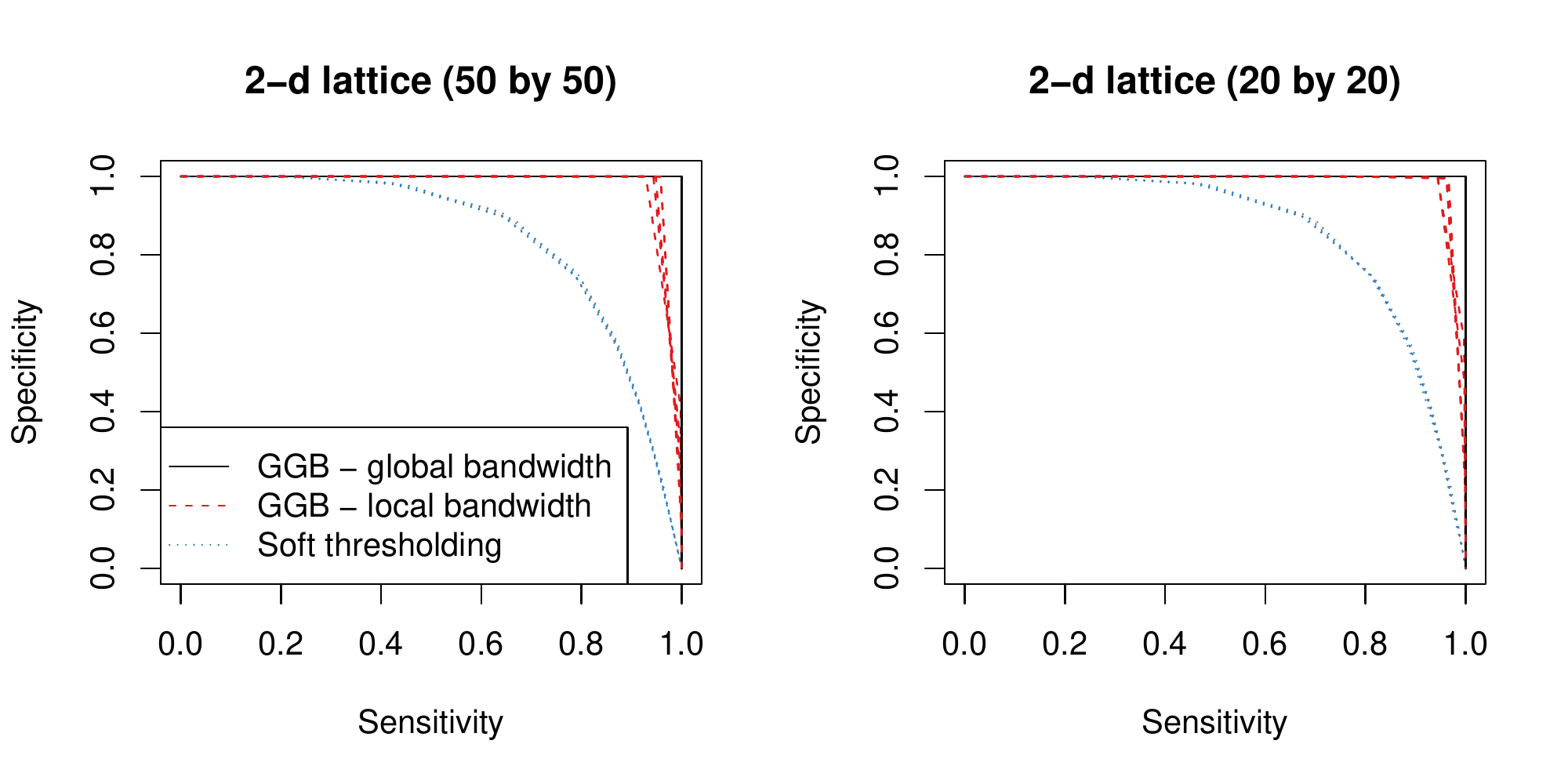}
  \caption{ROC curves for three methods when $p=400$ and $p=2,500$.}
  \label{fig:largerp-roc}
\end{figure}

\bibliography{refs}

\end{document}